\documentclass[11pt,draftclsnofoot,peerreviewca,letterpaper,onecolumn]{IEEEtran}
\usepackage{cite}
\usepackage{graphicx,color,epsfig,rotating}
\usepackage{amsfonts,amsmath,amssymb}
\usepackage{algorithm,algorithmic}
\usepackage{subfigure}

\long\def\comment#1{}

\newfont{\bbb}{msbm10 scaled 700}

\newfont{\bb}{msbm10 scaled 1100}
\newcommand{\CC}{\mbox{\bb C}}
\newcommand{\PP}{\mbox{\bb P}}
\newcommand{\RR}{\mbox{\bb R}}

\newcommand{\EE}{\mbox{\bb E}}

\newcommand{\av}{{\bf a}}
\newcommand{\bv}{{\bf b}}

\newcommand{\ev}{{\bf e}}

\newcommand{\qv}{{\bf q}}
\newcommand{\rv}{{\bf r}}
\newcommand{\sv}{{\bf s}}

\newcommand{\wv}{{\bf w}}
\newcommand{\vv}{{\bf v}}
\newcommand{\xv}{{\bf x}}
\newcommand{\yv}{{\bf y}}
\newcommand{\zv}{{\bf z}}
\newcommand{\zerov}{{\bf 0}}
\newcommand{\onev}{{\bf 1}}

\newcommand{\Am}{{\bf A}}

\newcommand{\Cm}{{\bf C}}

\newcommand{\Em}{{\bf E}}

\newcommand{\Qm}{{\bf Q}}

\newcommand{\Sm}{{\bf S}}

\newcommand{\Wm}{{\bf W}}

\newcommand{\Xm}{{\bf X}}

\newcommand{\Ac}{{\cal A}}

\newcommand{\Cc}{{\cal C}}

\newcommand{\Ec}{{\cal E}}

\newcommand{\Kc}{{\cal K}}
\newcommand{\Lc}{{\cal L}}

\newcommand{\Nc}{{\cal N}}

\newcommand{\Pc}{{\cal P}}

\newcommand{\Rc}{{\cal R}}
\newcommand{\Sc}{{\cal S}}

\newcommand{\lambdav}{\hbox{\boldmath$\lambda$}}

\newcommand{\omegav}{\hbox{\boldmath$\omega$}}

\renewcommand{\arg}{{\hbox{arg}}}

\newcommand{\SNR}{{\sf snr}}

\newcommand{\eqdef}{\stackrel{\Delta}{=}}

\newcommand{\transp}{{\sf T}}

\newtheorem{theorem}{Theorem}
\newtheorem{lemma}{Lemma}

\newtheorem{proof}{Proof}

\begin{document}

\title{Energy-Delay Tradeoff and Dynamic Sleep Switching 
for Bluetooth-Like Body-Area Sensor Networks}

\author{Eric Rebeiz,~\IEEEmembership{Student Member,~IEEE},
Giuseppe Caire,~\IEEEmembership{Fellow,~IEEE},\\
    and Andreas F. Molisch,~\IEEEmembership{Fellow,~IEEE}
\thanks{E. Rebeiz is with the Department of Electrical Engineering, University of California Los Angeles, Los Angeles CA  (e-mail:  rebeiz@ee.ucla.edu).}
\thanks{G. Caire and A. F. Molisch are with the Department of Electrical Engineering,
University of Southern California, Los Angeles, CA (e-mail: caire, molisch@usc.edu).}
}

\maketitle

\begin{abstract}
Wireless technology enables novel approaches to healthcare, in particular the remote monitoring 
of vital signs and other parameters indicative of people's health. 
This paper considers a system scenario relevant to such applications, where a smart-phone acts as a data-collecting hub, 
gathering data from a number of wireless-capable body sensors, 
and relaying them to a healthcare provider host through standard existing cellular networks. 
Delay of critical data and sensors' energy efficiency are both relevant and conflicting issues.
Therefore, it is important to operate the wireless body-area sensor network at some desired point 
close to the optimal energy-delay  tradeoff curve.  
This tradeoff curve is a function of the employed physical-layer protocol: in particular,
it depends on the multiple-access scheme and on the coding and modulation schemes available. 
In this work, we consider a protocol closely inspired by the widely-used Bluetooth standard.  
First, we consider the calculation of the minimum energy function, i.e., the minimum sum energy per symbol
that guarantees the stability of all transmission queues in the network.  
Then, we apply the general theory developed by Neely to develop a dynamic scheduling policy 
that approaches the optimal energy-delay tradeoff for the network at hand. 
Finally, we examine the queue dynamics and propose a novel policy that adaptively switches between connected and disconnected (sleeping) 
modes.   We demonstrate that the proposed policy can achieve significant gains in the realistic case where the 
control ``NULL'' packets necessary to maintain the connection alive, have a non-zero energy cost, and the data arrival statistics corresponding 
to the sensed physical process are bursty. 
\end{abstract}

\newpage

\section{Introduction} \label{intro}

Wireless Body Area Networks (WBAN) allow wireless connectivity among devices operating within very close proximity to the human body 
(see \cite{kohno2008,surveywban} and references therein). 
A commonly proposed architecture for WBAN considers sensors and/or actuators with limited power and buffering capacity, 
communicating with a body-wearable ``local hub'' node with significantly larger power, buffering and 
processing capability.  
The local hub acts as the WBAN controller, data fusion center, and gateway for external connectivity.
Smartphones are well suited to serve as such hubs because of their large storage capacity and processing power. 
Furthermore, most people already live with one of these phones in their pocket almost continuously, 
and naturally recharge their batteries for reasons that go beyond medical monitoring. 
The connection from the hub to the outside world poses no challenges since it makes use of a 
Wireless Local Area Network (WLAN) or 3G/LTE cellular data connection, either of which is widely available and provides more 
than sufficient data rate for the purpose of health monitoring  \cite{Molisch_2005_book}. 
Therefore, this paper focuses on the WBAN segment of the system, comprised of the  sensors and the local hub. 
In particular, since the local hub has a large battery and  it is easily recharged, we focus on the 
energy-delay tradeoff of the communication {\em from sensors to the hub}, although the WBAN protocol is bi-directional 
and involves the transmission of control packets also from the hub to the sensors as well.  
 
The sensing signals of interest for healthcare applications require data rates that range from some tens of 
bit/s to  hundreds of kbit/s, as shown in Table \ref{table} (see \cite{body2}).
The ``arrival process'' modeling the sensing operations may range from very impulsive to deterministic. 
For example, an ElectroEncephaloGraphy (EEG) signal triggered by an epileptic seizure may be silent for a long time and then suddenly produce a very large amount  of data within a short time period. In contrast, an ElectroCardioGraphy (ECG) signal for monitoring of 
physical activity in obese children \cite{knowme} produces a steady data rate. 

Given the random, possibly impulsive, nature of the sensor-generated data, the sensor node must be equipped 
with a transmission queue for local buffering. 
The rate and power allocated by the scheduler to the sensor-to-hub link  (uplink) must ensure 
stability of the queue  to avoid data loss resulting from buffer overflow. 
On the other hand, the energy efficiency of the sensor nodes is critically important since the sensor's battery may not be easily replaced 
or recharged.  Minimum energy expenditure subject to queue stability yields the so-called 
{\em minimum energy function}  \cite{Neely-IT07} of the system, defined as the infimum of
the sum of the sensors' average transmit energy per symbol,\footnote{In this paper we refer to ``energy'' to indicate average transmit energy per 
symbol, which corresponds to the transmit power, up to a fixed proportionality factor that depends on the signaling baud rate (symbols per second).}  
subject to stability of all queues. 
The transmission delay can be reduced by spending more than the minimum average energy per symbol.
The problem of optimal energy-delay tradeoff  was addressed by Berry and Gallager \cite{Berry-Gallager-IT02} for the single-user case 
and significantly extended by Neely \cite{Neely-IT07} to the $K \geq 1$ users case. 
In our case, we remark that smartphones have a built-in WBAN capability for short-range communication with other devices, currently based 
on the Bluetooth protocol \cite{bluetooth} (a possible rival protocol, IEEE 802.15.6, is in the process of being 
standardized \cite{802_15_6_draft}). Therefore, it is likely that the first generation of such WBAN applications will exploit Bluetooth. 

The main contribution of this paper consists of the application of the theoretical framework of \cite{Neely-IT07} to a
simplified, but realistic, ``Bluetooth-like'' protocol comprising of a finite number of possible physical layer (PHY) 
coding and modulation schemes (referred to in the following as {\em PHY modes}), 
each of which is characterized by a fixed data rate and given Packet Error Rate (PER) function of the received SNR. 
We start by addressing the calculation of the minimum energy function. For this problem, we provide an
exact solution method based on Lagrangian duality in the case where the PER functions, for each PHY mode, 
are indicator functions (PER equal to 0 for received SNR above some mode-dependent threshold, 
and equal to 1 otherwise). In the general case of smooth PER functions,  
we propose an achievable upper bound and a lower bound on the minimum energy function. The upper bound is based on
finding an appropriate feasible energy allocation policy, while the lower bound is obtained by finding a piecewise linear
upper bound on the effective rate function resulting from the actual (packet error-prone) PHY modes.
Next, we apply the general dynamic scheduling policy of \cite{Neely-IT07} to the case of a discrete set 
of non-ideal error-prone PHY modes,  and we provide a near-optimal scheduling rule that operates close  
to the optimum energy-delay tradeoff.  Finally, we observe that in the Bluetooth protocol the polled sensor needs to send NULL packets even when there 
is no data to transmit, in order to maintain the connection alive. In practice, these packets require some non-zero energy cost. 
Therefore,  it may be convenient to release the connection and switch to sleeping mode in the case of a very bursty arrival process. 
Based on this observation, we develop a novel improved policy that switches to sleeping mode whenever the power necessary 
to release the connection  and re-establish it at a later time is smaller than the conditionally expected power required 
to keep the connection alive, given the current  system state.  This new policy goes beyond the framework of  \cite{Neely-IT07}, and
provides significant power savings for practically relevant system parameters. 


The paper is organized as follows. Section \ref{bluetooth-like} defines a protocol closely inspired by Bluetooth, 
with some necessary simplifications in order to obtain an analytically tractable, yet significant, problem. 
In Section \ref{K-user}, we consider the general case of a piconet comprising $K$ sensors and a hub. 
The calculation of the minimum energy function is treated in Section \ref{min-pow-K}. 
In Section \ref{opt-scheduling}, 
we apply the  near-optimal  dynamic scheduling framework of \cite{Neely-IT07} to our setting, and
in Section \ref{switch} we discuss the behavior of the scheduling algorithm in the 
case of bursty arrivals, provide a new dynamic sleep switching policy and demonstrate through numerical examples 
its gain over the conventional policy that does not switch to sleep mode.
Conclusions are pointed out in Section \ref{conclusions}.

\section{A Bluetooth-like protocol} \label{bluetooth-like}

The PHY of the enhanced Bluetooth protocol \cite{bluetooth}  supports 3 data rates. 
The {\em basic rate} corresponds to a 1Mbps raw data rate using GFSK modulation \cite{Molisch_2005_book}. The {\em enhanced rate} 
modes use the $\pi/4$-DQPSK and 8-DPSK \cite{Molisch_2005_book}, with Gray bit-labeling, to achieve 2Mbps and 3Mbps of raw data rates. 
Since the signal bandwidth is approximately $1$ MHz, these correspond to spectral efficiencies 
$R_1 = 1, R_2 = 2$ and $R_3 = 3$ bit/s/Hz, respectively. 
The Bit Error Rates (BERs) of each of the modulation schemes in Additive White Gaussian Noise (AWGN) are given by the well-known
exact or tightly approximated expressions below \cite{millerdqpsk,proakis-book}:
\begin{eqnarray}
\epsilon_1(\SNR) &=& Q_{1}(a,b) - \frac{1}{2} \exp\left (-\frac{a^2+b^2}{2}\right )I_{0}(ab) \nonumber\\
\epsilon_2(\SNR) & \simeq & Q\left ( \sqrt{\SNR(2-\sqrt{2})} \right ) \nonumber\\
\epsilon_3(\SNR) & \simeq & \frac{2}{3}Q \left ( \sqrt{\SNR} \left (\sqrt{1+\sin(\pi/8)} - \sqrt{1-\sin(\pi/8)} \right ) \right ) ,
\end{eqnarray}
where $\epsilon_i(\SNR)$ denotes the BER for modulation scheme $i$ as a function of the {\em received} 
Signal-to-Noise Ratio (SNR), denoted by $\SNR$,  $Q_1(a,b)$ and $I_0(ab)$ denote the Marcum function and 
the zero-order modified Bessel function of the first kind respectively \cite{proakis-book}, and $a$ and $b$ are given by
\begin{equation} \label{ab}
a, b = \sqrt{\frac{\SNR}{2}\left (1\pm\sqrt{1-\left (\frac{\sin(2\pi h)}{2\pi h} \right )^2}\right ) }, 
\end{equation}
with $h=0.29$. 

Data packets consist of an Access Code (AC) that identifies the piconet, a Header (H) with various control information, 
and payload Data (D). AC and H are always transmitted using the basic rate, while D can be transmitted using any 
of the above modulations. As implemented in most commercial chipsets, we consider hard detection. 
Successful reception of the payload data can only occur after AC and H have been decoded successfully. 
For a given $\SNR$, constant over the transmission of a packet, the probability of successful decoding of the payload data is given by $P_{\rm D}(\SNR) = P_{{\rm A}_s}(\SNR) P_{{\rm H}_s}(\SNR) P_{{\rm D}_s}(\SNR)$, i.e., by the product of the individual probability of success of the three detection phases.

The AC is composed of 72 bits, 64 of which form a synchronization word with minimum Hamming 
distance $d_H=14$ between ACs of different piconets. 
If the Hamming distance between the detected and the expected synchronization words is not larger than 
the correlator margin $\rho$, 
the correlator at the receiving node triggers the reception of the packet; otherwise the packet is discarded. Thus, 
\begin{equation}
P_{{\rm A}_s}(\SNR) = \sum_{k=0}^{\rho}{{64 \choose k} \epsilon_1(\SNR)^k(1-\epsilon_1(\SNR))^{64-k}}.
\end{equation}
Notice that $\rho$ is not specified by the standard and can be set by the manufacturer.  
In our simulations, we assumed $\rho=6$ \cite{Zanella-bt-efficiency}.  

The header is composed of 18 bits, each encoded by a $(3,1)$ repetition code. 
The probability of a header being successfully decoded is given by
\begin{equation}
P_{{\rm H}_s}(\SNR) =  \left ( (1-\epsilon_1(\SNR))^3+3\epsilon_1(\SNR)(1-\epsilon_1(\SNR))^2 \right )^{18}.
\end{equation}
The header also contains a checksum for error detection, used to verify header integrity after decoding. The checksum is designed 
so that the probability of undetected error is sufficiently small, such that we can safely assumed that decoding errors can be revealed with probability 1.

The payload can either be encoded using a rate 2/3 extended Hamming code or be transmitted uncoded. 
In particular, the enhanced rates  are uncoded so that
\begin{equation}
P_{{\rm D}_s}(\SNR) = \left ( 1-\epsilon_i(\SNR) \right )^{B_i}, 
\end{equation}
where $i$ is the chosen rate, and $B_i$ is the number of payload  bits, as given in Table \ref{PacketSlots}. 
The payload can be transmitted over 1, 3, or 5 time slots.  When more than one time slot is allocated, the protocol 
overhead (AC and H) is only  included in the first time slot. Therefore, a larger number of slots yields lower relative protocol 
overhead, but also a larger probability of packet error. 

When a node has nothing to send, but the connection to the piconet is to be maintained, ``NULL'' control packets are used. These 
have AC and H as described above, but contain no payload. In this case, the probability of successful NULL packet detection is
simply given by $P_{\rm NULL}(\SNR) = P_{{\rm A}_s}(\SNR) P_{{\rm H}_s}(\SNR)$. The minimum received SNR required to stay 
connected to the piconet shall be denoted by $\SNR_0$. For example, $\SNR_0 = 8$ dB yields $P_{\rm NULL}(\SNR_0) \approx 0.95$.

It is also interesting to notice that the basic rate is not competitive with respect to the two enhanced rates, due to the poor BER performance of 
non-coherent GFSK. Therefore, in our results we shall consider only enhanced data rates for the payload. 
Fig. \ref{PERfig} shows  probability of successful packet detection vs. the received SNR for the 
enhanced rate modes for different payload lengths. 

Bluetooth supports both Asynchronous (ACL) and Synchronous (SCO) types of data transfer. 
Correspondingly, different logical link control and transport layers have been defined in the system architecture specification. 
In the system considered here, the hub acts as the master of the Bluetooth piconet 
and polls each sensor according to some scheduling scheme, to be investigated later. 
The sensors (slaves) respond with an information packet on the next slot(s). 
ACL transfer supports asymmetric transmission rates, flexibility of non-periodic polling, 
re-transmission and the ability to maintain piconets with several slaves. Therefore, 
ACL is better suited for the application in mind.  When a device has an active ACL link, it is said to be in 
the {\em connected state},  and can communicate with the master whenever polled.

In this work we assume that a connected slave device responds to polls from the master by sending data packets at given 
power and rate, according to the PHY layer options described before, or NULL control packets if the slave has nothing to send; 
we emphasize that NULL packets must be received correctly with sufficiently high probability; 
otherwise the network loses synchronism. 
While previous theoretical work \cite{Berry-Gallager-IT02,Neely-IT07} assumed that there is no power cost when no data are to be transmitted, 
the NULL packet transmission does incur a power cost -- a fact that we shall specifically take into account in this paper. 
As an alternative to NULL packet transmission, a device can be disconnected from the network if it has nothing to send for 
a long time. In this case, an (implementation-specific) power cost is incurred in the 
disconnection and later inquiry, paging and re-connection operations. In our simulations  
we have considered a reconnection cost equal to some integer multiple $\tau \geq 1$ of the cost of sending a 
NULL packet,  assuming that the reconnection operation requires $\tau$ NULL control packets to be sent by the sensor
(see Section \ref{switch}).


\section{Problem set-up} \label{K-user}

Consider the system shown in Fig.~\ref{system1}, with $K$ transmitters (sensors) and one receiver (the local hub). 
Time is divided into slots, each of which comprises $N  \approx W_s T_s$ complex symbols (channel uses), where 
$W_s$ is the slot bandwidth and $T_s$ is the slot duration  \cite[Ch. 8]{gallager}.  Rates are expressed in bit/s/Hz.

The complex discrete-time base-band equivalent channel model for sensor $k$ transmission is given by 
\begin{equation} \label{phy-channel}
\yv(t) = \sum_{k=1}^K \sqrt{S_k(t)} \xv_k(t) + \zv(t),
\end{equation}
where $\yv(t), \xv_k(t), \zv(t) \in \CC^N$ denote the received signal, the transmit codeword 
and the AWGN in slot $t$, and where $S_k(t)$ denotes the channel fading state, assumed to be constant in each slot, 
and changing randomly from slot to slot according to some i.i.d. process. 
This block-fading channel model is motivated by the fact that, in typical WBANs, 
the coherence time and coherence bandwidth are of the order of $10$ ms and 
$10$ MHz \cite{Molisch_et_al_2006_AP,Stability-of-narrowband,802_15_6_draft},  
and thus they are respectively larger than the packet duration and of the signal bandwidth. 
Furthermore, a Bluetooth-like protocol makes use of slow frequency hopping \cite{bluetooth}, 
for which different packets transmitted at sufficiently far apart frequencies 
undergo independent channel states.~\footnote{Notice that hopping is done on a per-packet basis, not a per-timeslot basis.} 
The average transmit energy per symbol of sensor $k$ in slot $t$ is defined as
\begin{equation}
E_k(t) = \frac{1}{N} \EE[\| \xv_k(t) \|^2 ],
\end{equation}
and the AWGN is normalized so that its components are i.i.d.  $\sim \Cc\Nc(0,1)$ 
(complex circularly symmetric Gaussian with unit variance and zero mean).
As in Bluetooth, we consider Time-Division Multiple Access (TDMA) where the hub schedules 
a single transmitter to be active over  any slot $t$. Hence, in the channel model (\ref{phy-channel}) the sum over $k$ contains effectively only 
one term for $k = k(t)$,  where $k(t)$ denotes the index of the scheduled transmitter at slot time $t$. 
Correspondingly, $E_k(t) = 0$ for all $k \neq k(t)$. 

For the scheduled sensor $k(t)$, we let $\ell(t)$ and $E_{k(t)}(t)$ denote the 
PHY mode and the corresponding transmit energy chosen by the scheduling policy, respectively. 
Consistent with the Bluetooth-like protocol outlined in Section \ref{bluetooth-like}, 
we assume that a finite number $L \geq 1$ of possible PHY modes are available, where each mode 
$\ell$ is characterized by rate $R_\ell$ and a probability of successful packet detection, function of the 
received SNR,  denoted by  $P_\ell(\SNR)$.  For example, in the case of Bluetooth, we have $P_\ell(\SNR) = 
P_{{\rm D}}(\SNR)$  for the corresponding enhanced rate mode and packet length. 
The received SNR for sensor $k(t)$ is given by the product  $\SNR_{k(t)} = E_{k(t)}(t) S_{k(t)}(t)$. 
Also, we assume that the PHY has an {\em idle} mode denoted by $\ell = 0$, with transmission rate $R_0 = 0$ and minimum 
received SNR equal to some $\SNR_0 \geq 0$, that depends on the reliability with which
the NULL control packets need to be detected.  Without loss of generality, we assume $R_0 = 0 < R_1 < \cdots  < R_L$.  
Letting $\SNR_0 = 0$, we can consider the ideal case where NULL slots do not incur any power cost, 
as considered in \cite{Berry-Gallager-IT02,Neely-IT07}. 

We define the conditional successful decoding event $\Ac(\ell, e, s)$ for 
PHY mode $\ell$ with transmit energy $e$ and
channel fading state $s$ as:
\begin{equation} \label{successful} 
\Ac(\ell, e, s) = 
\left \{ \mbox{successful decoding at slot $t$} | \ell(t) = \ell,  E_{k(t)}(t) = e, S_{k(t)}(t) = s \right \}.
\end{equation}
By definition, we have $P_\ell(es) = \EE[ 1\{ \Ac(\ell, e, s) \} ]$.~\footnote{$1\{\Ac\}$ denotes the indicator function of an event 
$\Ac$ in the underlying probability space.}
The data measured by the $k$-th sensor form an i.i.d. arrival process $A_k(t)$, with arrival rate $\EE[A_k(t)] = \lambda_k$ bit/s/Hz. Arrival processes are assumed to be bounded with probability 1, i.e., there exists some constant $A_{\max}$ such that $\PP(A_k(t) \in [0, A_{\max}]) = 1$ for all $k$. 
Each sensor $k$ has a transmission queue with backlog $Q_k(t)$, also expressed in bit/s/Hz after a suitable normalization. 
We assume that some error detection and logical link control mechanism 
keeps the information bits in the transmit queue for later transmission if a packet error occurs.
Then, the transmission queues evolve according to the stochastic difference equations 
\begin{equation} \label{buffer-evo}
Q_k(t+1) = \left [ Q_k(t) -  R_{\ell(t)} \times 1\left \{ \Ac(\ell(t), E_k(t), S_k(t)) \right \} \times 1\{ k(t) = k \} \right ]_+ + A_k(t), 
\end{equation}
for $k = 1,\ldots, K$,  where $[ x]_+ \eqdef \max\{x, 0\}$.

By Little's theorem \cite{gallager-bertsekas}, the average delay for the system in Fig.~\ref{system1} 
is given by $\overline{D}_k = \overline{Q}_k/\lambda_k$, where 
$\overline{Q}_k = \lim_{t \rightarrow \infty} \frac{1}{t} \sum_{\tau=0}^{t-1} \EE[Q_k(\tau)]$ is the limit of the average 
queue buffer size.  When $\overline{Q}_k  < \infty$ for all $k$, the queues are said to be {\em strongly stable} 
\cite{Georgiadis-Neely-Tassiulas-2004}.  Under our system assumptions,  a necessary and sufficient condition for strong stability is
that the average service rate for each sensor $k$ must be larger than its arrival rate, i.e., 
\begin{equation} \label{strong-stab}
\liminf_{t \rightarrow \infty} \frac{1}{t} \sum_{\tau =0}^{t-1} \EE \left [ R_{\ell(\tau)} \times 
1\left \{ \Ac(\ell(\tau), E_k(\tau), S_k(\tau)) \right \} \times 1\{ k(\tau) = k\} \right ] > \lambda_k, \;\;\; \forall \;\; k = 1,\ldots, K. 
\end{equation}
A scheduling policy achieving finite $\overline{Q} = \frac{1}{K} \sum_{k=1}^K \overline{Q}_k$ (i.e., finite average delay) 
is referred  to as a {\em stability policy}, and the closure of the convex hull of all arrival rate vectors $\lambdav = (\lambda_1, \ldots, \lambda_K)$ such that there exists a stability policy is referred to as the {\em stability region}, and it is indicated here 
by $\Lambda$, consistent with the notation of \cite{Neely-IT06,Neely-IT07}.\footnote{This is referred to as ``capacity region'' in \cite{Neely-IT06}, but we prefer to reserve this term to indicate  the Shannon capacity region of a multi-terminal network, which may or may not coincide with the stability region, depending on the assumptions.}
 
The average transmit energy per symbol for a policy $\pi(t) = (k(t), \ell(t), E_{k(t)}(t))$ is  given by
\begin{equation} \label{average-pow}
\overline{E} = \lim_{t \rightarrow \infty} \frac{1}{t} \sum_{\tau=0}^{t-1} \sum_{k=1}^K \EE \left [ E_k(\tau) \times 1\{ k(\tau) = k\} \right ]. 
\end{equation}
The minimum energy function, indicated by $\Phi(\lambdav)$ consistent with \cite{Neely-IT06,Neely-IT07}, 
is defined as the minimum average transmit energy per symbol required for stability under 
the arrival rate vector $\lambdav \in \Lambda$. Using \cite[Theorem 1]{Neely-IT06}, which can be applied in 
our case almost verbatim,  we have that $\Phi(\lambdav)$ is achieved by a randomized stationary policy, i.e., by a time-invariant 
random mapping of the channel state vector $\Sm(t) = (S_1(t), \ldots, S_K(t))$ into a scheduling decision. 
Formally, such a policy is defined 
by the conditional probability distribution
\begin{equation} \label{policy1} 
P_{\Kc_\pi, \Lc_\pi, \Ec_\pi |\Sm}(k,\ell, e | \sv)  =  \PP(\Kc_\pi = k, \Lc_\pi = \ell, \Ec_\pi = e | \Sm = \sv), 
\end{equation}
where $\Sm$ denotes the channel state vector, 
and $\Kc_\pi, \Lc_\pi$ and $\Ec_\pi$ denote random variables
defined over $\{1, \ldots, K\}$, $\{0, \ldots, L\}$ and $\RR_+$ (the set of non-negative reals), indicating the 
index of the scheduled sensor, the index of the selected PHY mode, and the corresponding transmit energy, respectively. 
Such a policy works as follows: at given slot time $t$, sensor $k(t)$, is scheduled 
with PHY mode $\ell(t)$ and transmits with energy $E_{k(t)}(t)$ with probability
$P_{\Kc_\pi, \Lc_\pi, \Ec_\pi |\Sm}(k(t),\ell(t), E_{k(t)}(t) | \Sm(t))$.   
In Section \ref{min-pow-K} we discuss the calculation of $\Phi(\lambdav)$ for the system at hand. 

More generally, the energy-delay tradeoff region of the system for given arrival rates $\lambdav$ 
is defined as the set of points  $(\overline{E}, \overline{D})$ in the energy-delay plane for which there exists 
a scheduling  policy $\pi$ that achieves an average delay $\frac{1}{K} \sum_{k=1}^K \overline{D}_k \leq \overline{D}$ 
with average energy per symbol equal to $\overline{E}$. 
The Pareto boundary of this region is the {\em optimal energy-delay tradeoff} \cite{Berry-Gallager-IT02,Neely-IT07}. 
In general, the optimal energy-delay tradeoff is  a non-increasing convex curve with the horizontal asymptote 
$\overline{E} \downarrow \Phi(\lambdav)$ as $\overline{D} \rightarrow \infty$. 
It is known \cite{Berry-Gallager-IT02,Neely-IT07} that scheduling policies that depend on the channel state only (i.e., which disregard the 
queue buffer states) cannot achieve optimal energy-delay tradeoff.  Therefore, as in \cite{Neely-IT07}, we are interested in the class of randomized stationary 
policies $\pi$,  mapping the channel state vector $\Sm(t)$ and the buffer state vector $\Qm(t) = (Q_1(t), \ldots, Q_K(t))$ into a 
scheduling decision. Formally, such a policy is  defined by the conditional probability distribution
\begin{equation} \label{policy2} 
P_{\Kc_\pi, \Lc_\pi, \Ec_\pi |\Sm, \Qm}(k,\ell, e | \sv,\qv)  =  \PP(\Kc_\pi = k, \Lc_\pi = \ell, \Ec_\pi = e | \Sm = \sv,  \Qm = \qv), 
\end{equation}
where all random variables are as defined above, and  $\Qm$ denotes the buffer state vector. 
Such a policy works as follows: at given slot time $t$, sensor $k(t)$, is scheduled 
with PHY mode $\ell(t)$ and transmits with energy $E_{k(t)}(t)$ with probability
$P_{\Kc_\pi, \Lc_\pi, \Ec_\pi |\Sm, \Qm}(k(t),\ell(t), E_{k(t)}(t) | \Sm(t),\Qm(t))$.   
In Section \ref{opt-scheduling} we apply the {\em dynamic scheduling} framework of \cite{Neely-IT07} 
to the system at hand and obtain a family of policies that can operate closely to any desired point of the optimal energy-delay tradeoff curve 
by appropriately tuning  a control parameter.

\section{Minimum average energy subject to stability constraints} \label{min-pow-K}

Let $\Sm = (S_1, \ldots, S_K)$ denote a random vector with the same first order marginal probability cumulative distribution function (cdf) 
$F_S(\sv)$ of the channel  state $\Sm(t)$.  As anticipated before, an underlying system assumption is that the network must be designed 
in order to guarantee successful reception of the NULL packets with sufficiently high probability. 
In practice, in the case of protocol errors, some paging and synchronization mechanism allows to restart 
operations.~\footnote{For example,  anyone using a Bluetooth wireless mouse has had the experience of occasional ``mouse disconnected'' error messages appearing on the  laptop screen.}  
However, for simplicity, protocol errors are not taken into account and $\SNR_0$ is fixed such that the $P_{\rm NULL}(\SNR_0)$ is sufficiently large.  
In particular, in order to guarantee $\SNR_{k(t)} \geq \SNR_0 > 0$ with probability 1, we assume that $F_S(\cdot)$ is such 
that $\EE[1/S_k] < \infty$ for all $k$. 
More details on the channel state cdf used in our numerical results, relevant for the WBAN scenario, is given at the end of this section. For the rest, 
the minimum energy calculation discussed here is general, and holds for any continuous, discrete or mixed-type \cite{grimmett} fading distribution. 


Finding $\Phi(\lambdav)$ consists of solving for the conditional probability distribution (\ref{policy1}) such that
the corresponding randomized stationary policy minimizes $\overline{E}$ subject to stability of all queues. 
It is useful to separate the discrete and continuous parts of the conditional probability distribution (\ref{policy1}) as
\begin{equation}
P_{\Kc_\pi, \Lc_\pi, \Ec_\pi |\Sm}(k,\ell, e | \sv) = p_{\Ec_\pi |\Kc_\pi, \Lc_\pi, \Sm}( e | k,\ell,\sv) P_{\Kc_\pi, \Lc_\pi |\Sm}(k,\ell | \sv).
\end{equation}
The discrete conditional probability distribution $P_{\Kc_\pi, \Lc_\pi |\Sm}(k,\ell | \sv)$ is represented by a 
matrix-valued function $\Wm(\sv) \in [0,1]^{K \times (L+1)}$ with $(k,\ell)$ element $w_{k,\ell}(\sv)$, such that
\begin{equation}
P_{\Kc_\pi, \Lc_\pi |\Sm}(k,\ell | \sv) = w_{k,\ell}(\sv), \;\;\;\; \forall \;\; \sv \in \RR_+^K. 
\end{equation}
The continuous part corresponds to the energy allocation random variable 
$\Ec_{k,\ell}(\sv)$, conditionally distributed according to $p_{\Ec_\pi |\Kc_\pi, \Lc_\pi, \Sm}( e | k,\ell,\sv)$. 
Using the stability condition (\ref{strong-stab}), $\Phi(\lambdav)$ is given by the  value of the optimization problem
\begin{eqnarray} 
\mbox{minimize} & &  \EE \left [ \sum_{k,\ell} w_{k,\ell}(\Sm) \Ec_{k, \ell}(\Sm)  \right ] \label{K-user-Phi-smooth} \\
\mbox{subject to} & & \EE \left [ \sum_{\ell} w_{k,\ell}(\Sm) R_\ell \times 1\{ \Ac(\ell, \Ec_{k,\ell}(\Sm), S_k) \}   \right ] \geq \lambda_k, \;\;\; \forall \; k = 1,\ldots, K \nonumber \\
& & \onev^\transp \Wm(\sv) \onev = 1, \;\;\; \forall \;\; \sv \in \RR_+^K \nonumber \\
& & \PP( s_k \Ec_{k,\ell}(\sv) \geq \SNR_0) = 1, \;\;\; \forall \;\; \sv \in \RR_+^K, \;\; k = 1,\ldots,K, \;\;\; \ell = 0,\ldots, L, \nonumber \\
& & w_{k,\ell}(\sv) \geq 0, \;\;\; \forall \;\; \sv \in \RR_+^K, \;\; k = 1,\ldots,K, \;\;\; \ell = 0,\ldots, L,  \nonumber 
\end{eqnarray}
where $\onev$ denotes the all-one vector, and where the conditional successful decoding event $\Ac(\ell,e,s)$ 
is defined in (\ref{successful}).  Using iterated expectation, we can write 
\begin{eqnarray} \label{constraint}
\EE \left [ \sum_{\ell} w_{k,\ell}(\Sm) R_\ell \times 1\{ \Ac(\ell, \Ec_{k,\ell}(\Sm), S_k) \}   \right ] & = & 
\EE \left [ \EE \left [ \left . \sum_{\ell} w_{k,\ell}(\Sm) R_\ell \times 1\{ \Ac(\ell, \Ec_{k,\ell}(\Sm), S_k) \}  \right | \Sm, \Ec_{k,\ell}(\Sm) \right ] \right ] \nonumber \\
& = & 
\EE \left [ \sum_{\ell} w_{k,\ell}(\Sm) R_\ell \times \EE\Big [ \left . 1\{ \Ac(\ell, \Ec_{k,\ell}(\Sm), S_k) \}  \right | \Sm, \Ec_{k,\ell}(\Sm) \Big ] \right ] \nonumber \\
& = & 
\EE \left [ \sum_{\ell} w_{k,\ell}(\Sm) R_\ell P_\ell (\Ec_{k,\ell}(\Sm) S_k)  \right ]
\end{eqnarray}
Using (\ref{constraint}) into the rate constraint in (\ref{K-user-Phi-smooth}), we rewrite our problem as
\begin{eqnarray} 
\mbox{minimize} & &  \EE \left [ \sum_{k,\ell} w_{k,\ell}(\Sm) \Ec_{k, \ell}(\Sm)  \right ] \label{K-user-Phi-smooth1} \\
\mbox{subject to} & & \EE \left [ \sum_{\ell} w_{k,\ell}(\Sm) R_\ell P_\ell (\Ec_{k,\ell}(\Sm) S_k)   \right ] \geq \lambda_k, \;\;\; \forall \; k = 1,\ldots, K \nonumber \\
& & \onev^\transp \Wm(\sv) \onev = 1, \;\;\; \forall \;\; \sv \in \RR_+^K \nonumber \\
& & \PP( s_k \Ec_{k,\ell}(\sv) \geq \SNR_0) = 1, \;\;\; \forall \;\; \sv \in \RR_+^K, \;\; k = 1,\ldots,K, \;\;\; \ell = 0,\ldots, L, \nonumber \\
& & w_{k,\ell}(\sv) \geq 0, \;\;\; \forall \;\; \sv \in \RR_+^K, \;\; k = 1,\ldots,K, \;\;\; \ell = 0,\ldots, L.  \nonumber 
\end{eqnarray}
The form of problem (\ref{K-user-Phi-smooth1}) is appealing since
the conditional {\em effective rate} terms $R_\ell P_\ell (\Ec_{k,\ell}(\Sm) S_k)$ appearing 
in (\ref{constraint}) coincide with the average ``goodput'' (product of the transmitter rate times the probability of decoding success) 
corresponding to each PHY  mode $\ell$ with received SNR equal to $\Ec_{k,\ell}(\Sm) S_k$.  
For general smooth  functions $P_\ell(\SNR)$, finding the optimal 
stationary policy (i.e., the matrix $\Wm(\sv)$ and the energy conditional probability distribution $p_{\Ec_\pi |\Kc_\pi, \Lc_\pi, \Sm}( e | k,\ell,\sv)$)
is not easy. In fact, using Carath\'eodory theorem  and the approach of \cite{Neely-IT06}, the solution can be given as a convex 
combination (whose coefficients need to be optimized) of $K + 2$ pairs of matrices $\{\Wm_j(\sv), \Em_j(\sv) : j = 1,\ldots, K+2\}$
for each channel state $\sv \in \RR^K_+$  where $\Wm_j(\sv)$ are instances of $\Wm(\sv)$ and $\Em_j(\sv)$ contain instances of the realization 
of the random variable $\Ec_{k,\ell}(\sv)$ for all $k, \ell, \sv$ (see \cite[Appendix A]{Neely-IT06} for details).
This approach, however, does not lead to a viable computational method since it requires discretizing
the channel state space into a sufficiently fine grid of points $\Sc \subset \RR_+^K$, and solving the corresponding 
non-convex optimization problem in $2 LK (K+2) |\Sc|$ variables.   

In order to overcome this difficulty, we propose a method that efficiently finds an achievable sub-optimal 
solution $\overline{\Phi}(\lambdav)$ and, similarly, a lower bound $\underline{\Phi}(\lambdav) \leq \Phi(\lambdav)$. 
The achievable upper bound on the minimum energy function is constructive, and leads to a feasible scheduling policy
that achieves stability whenever $\overline{\Phi}(\lambdav) < \infty$. The lower bound allows us to 
precisely quantify the loss incurred by the simplified feasible policy. In order to prove upper bound, lower bound, 
and an efficient solution method, we need some auxiliary results given below
and proved in Appendix \ref{proofs}. 

Define the effective rate vs. SNR function 
\begin{equation} \label{mu} 
\mu(\SNR) = \max_{\ell = 0,\ldots, L} \left \{ R_\ell P_\ell(\SNR) \right \},
\end{equation}
with domain $\SNR \geq \SNR_0$. An example of such function is shown in Fig.~\ref{PERfig-example}, 
where $\mu(\SNR)$ is plotted vs $\SNR$ (in natural scale, not in dB) for the case of $\SNR_0 = 8$ dB, 
$\ell = 1$ and $\ell = 2$ corresponding to the 2dh3 and 3dh3 enhanced data rate modes of Bluetooth
with payload length equal to 3 slots, used throughout our numerical results (see Table. \ref{PacketSlots}). 

Let  $\overline{\mu}(\SNR)$ denote the upper concave envelope \cite{horst} of $\mu(\SNR)$, with the same support. 
By definition, for all $\ell$, we have
\[ R_\ell P_\ell (\SNR) \leq \mu(\SNR) \leq \overline{\mu}(\SNR), \;\;\; \forall \;\; \SNR \geq \SNR_0. \]
Hence, a relaxation of problem (\ref{K-user-Phi-smooth1}) is obtained as 
\begin{eqnarray} 
\mbox{minimize} & &  \EE \left [ \sum_{k,\ell} w_{k,\ell}(\Sm) \Ec_{k, \ell}(\Sm)  \right ] \label{K-user-Phi-smooth2} \\
\mbox{subject to} & & \EE \left [ \sum_{\ell} w_{k,\ell}(\Sm) \overline{\mu} (\Ec_{k,\ell}(\Sm) S_k)   \right ] \geq \lambda_k, \;\;\; \forall \; k = 1,\ldots, K \nonumber \\
& & \onev^\transp \Wm(\sv) \onev = 1, \;\;\; \forall \;\; \sv \in \RR_+^K \nonumber \\
& & \PP( s_k \Ec_{k,\ell}(\sv) \geq \SNR_0) = 1, \;\;\; \forall \;\; \sv \in \RR_+^K, \;\; k = 1,\ldots,K, \;\;\; \ell = 0,\ldots, L, \nonumber \\
& & w_{k,\ell}(\sv) \geq 0, \;\;\; \forall \;\; \sv \in \RR_+^K, \;\; k = 1,\ldots,K, \;\;\; \ell = 0,\ldots, L,  \nonumber 
\end{eqnarray}
We have:

\begin{lemma} \label{deterministic-E}
The relaxed problem (\ref{K-user-Phi-smooth2}) admits a solution where $\Ec_{k,\ell}(\sv)$ is a deterministic function
$e_{k,\ell}(\sv)$ of  $k,\ell$ and $\sv$. \hfill \QED
\end{lemma}

\noindent
The following result  identifies a case for which the relaxation is tight. 

\begin{lemma} \label{indicators}
Consider the case $P_\ell(\SNR) = 1\{ \SNR \geq a_\ell \}$, for some SNR thresholds 
$a_0 = \SNR_0 \leq a_1 \leq \cdots \leq a_L$,  and assume that the non-decreasing piecewise 
linear function joining the points $(a_\ell, R_\ell)$ for
$\ell = 0,\ldots, L+1$ is concave (note: for convenience we define $a_{L+1} = +\infty$ and $R_{L+1} = R_L$). 
Then, the relaxed problem (\ref{K-user-Phi-smooth2}) yields the same 
value as  the original problem (\ref{K-user-Phi-smooth1}), and this is obtained by the deterministic energy allocation function
$e_{k,\ell}(\sv) = a_\ell/s_k$. \hfill \QED
\end{lemma}

It is also immediate to observe that when the PHY mode success probabilities are given by indicator functions (as in the assumption of Lemma \ref{indicators}) but the piecewise linear function joining the points $(a_\ell, R_\ell)$ for $\ell = 0,\ldots, L+1$ is not concave,  then we can obtain an equivalent problem with a reduced number $L' < L$ of PHY modes, by retaining only the modes for which the corresponding point $(a_\ell, R_\ell)$ is on the boundary of the convex hull. 
In other words, PHY modes whose SNR threshold -- rate point $(a_\ell, R_\ell)$ is inside the convex hull of other modes
are irrelevant as far as the minimum energy problem is concerned, and can be dropped
without changing the system performance. Therefore, there is no loss of generality in the concavity assumption of
Lemma \ref{indicators}. 

Summarizing: we have identified one case, namely, the case where the PHY mode success probabilities 
are given by indicator functions $1\{ \SNR \geq a_\ell\}$ for some mode-dependent SNR thresholds $a_\ell$, for which
the relaxation obtained by replacing $\mu(\SNR)$ with its upper concave envelope $\overline{\mu}(\SNR)$, 
does not increase the value of the problem. Moreover, in this case the optimal energy allocation is
deterministic, given explicitly by  $e_{k,\ell}(\sv) = a_\ell/s_k$.  As we will see next, for this case we can find efficiently 
an exact solution. Since the case of concave increasing piecewise linear rate functions 
is easily handled,  it is natural to seek achievable upper bounds and lower bounds for the case of 
general $P_\ell(\SNR)$ by defining appropriate concave piecewise linear rate functions.

\subsection{Exact solution for the case $P_\ell(\SNR) = 1\{\SNR \geq a_\ell\}$}

In passing, we notice that this case is relevant for long packets and strong channel coding schemes. 
For example, assuming a family of capacity-achieving codes for the circularly symmetric AWGN channel \cite{Cover_Thomas_book}, and using the strong converse to the channel coding 
theorem,\footnote{The strong converse result states that if an ensemble of capacity achieving codes is used
at rates above the capacity, its block error probability converges to 1 as the block length increases.  This result is proved in great generality in \cite{verdu-han}.} 
we have that the successful decoding probability for rate $R_\ell$ behaves as an indicator function with SNR threshold $a_\ell = 2^{R_\ell} - 1$.  In practice, heavily coded systems (e.g., systems using powerful turbo or LDPC codes \cite{urbanke-book}, or concatenated convolutional and Reed-Solomon codes \cite{proakis-book}) are characterized by very sharp error probability waterfalls: as soon as the received SNR falls below  a certain code-dependent threshold, 
the post-decoding block error probability increases very rapidly.  In this case, approximating $P_\ell(\SNR)$ with an indicator function for an appropriately chosen SNR threshold
may be meaningful. 

By Lemmas \ref{deterministic-E} and \ref{indicators} and their consequences summarized before, we have that
in the case $P_\ell(\SNR) = 1\{\SNR \geq a_\ell\}$ the minimum energy function problem (\ref{K-user-Phi-smooth1})
is equivalent to:
\begin{eqnarray} 
\mbox{minimize} & &  \EE \left [ \bv^\transp(\Sm) \Wm(\Sm) \av  \right ] \label{K-user-Phi1} \\
\mbox{subject to} & & \EE \left [  \Wm(\Sm) \rv   \right ] \geq \lambdav, \nonumber \\
& & \onev^\transp \Wm(\sv) \onev = 1, \;\;\; \forall \;\; \sv \in \RR_+^K \nonumber \\
& & w_{k,\ell}(\sv)  \geq 0, \;\;\; \forall \;\; \sv \in \RR_+^K, \;\; k = 1,\ldots,K, \;\;\; \ell = 0,\ldots, L,  \nonumber
\end{eqnarray}
where we define $\bv(\sv) = (1/s_1, \ldots, 1/s_K)^\transp$, $\av = (a_0, \ldots, a_L)^\transp$ and
$\rv = (0, R_1, \ldots, R_L)^\transp$. 

Notice that (\ref{K-user-Phi1}) is an infinite-dimensional linear program, 
since both the objective function and the constraints are linear in the weighting function $\Wm(\sv)$.  
For notation convenience and later use, we let $r_\ell$ denote the $\ell$-th components of the rate vector $\rv$. 
We shall discuss the solution (\ref{K-user-Phi1}) allowing for some lack of rigor.  We can think of $\Wm(\sv)$ as an array of variables 
$w_{k,\ell}(\sv)$ for each $k,\ell$ and $\sv$, and replace expectations with  
the corresponding Stieltjes integrals \cite{grimmett}.
\footnote{A more rigorous approach consists of using variational calculus, but here we chose this more intuitive argument.} 
Then, the Lagrangian function for (\ref{K-user-Phi1}) is given by 
\begin{eqnarray} \label{lagrangian-K-user}
\Lc(\Wm, \omegav,\nu) & = & \int_{\RR_+^K}  \left [ \bv^\transp(\sv) \Wm(\sv) \av  - \omegav^\transp \left ( \Wm(\sv) \rv  - \lambdav \right ) 
+ \nu(\sv) \left ( \onev^\transp \Wm(\sv) \onev  -  1 \right ) \right ] \; dF_S(\sv), 
\end{eqnarray}
where $\omegav$ (non-negative vector) and $\nu(\sv)$ (non-negative scalar function) are 
dual variables.  By taking the partial derivative with respect to $w_{k,\ell}(\sv)$, for given $k,\ell$ and $\sv$,  
we find
\begin{eqnarray} \label{KKT2}
\frac{\partial \Lc}{\partial w_{k,\ell}(\sv)} & = & dF_S(\sv) \left [ \frac{a_\ell}{s_k}  - \omega_k r_\ell  + \nu(\sv) \right ]. 
\end{eqnarray}
Then, since (\ref{K-user-Phi1}) is a linear program in $\Wm(\sv)$,  
the solution must be found at the vertex of the domain of $\Wm(\sv)$ for which 
the inner product between $\Wm(\sv)$ and the vector of partial derivatives  is minimum. This corresponds to
choosing with probability 1 the sensor $k$ with PHY mode $\ell$ such that 
the corresponding term $\frac{\partial \Lc}{\partial w_{k,\ell}(\sv)}$ in (\ref{KKT2}) is minimum. Hence, we have proved the following result:

\begin{theorem} \label{multi-sensor-min-power}
The minimum of (\ref{lagrangian-K-user}) over $\Wm(\sv)$ for fixed $\omegav \geq 0$ and given  vectors $\av$ and $\rv$  is given by 
\begin{equation} \label{sol-K}
w_{k,\ell}(\sv) = \left \{ 
\begin{array}{ll}
1 & \mbox{if} \; \omega_k r_\ell  - \frac{a_\ell}{s_k}  \geq \omega_i r_j  - \frac{a_j}{s_i}  \; \forall \; (i,j) \\
0 & \mbox{otherwise} 
\end{array} \right . 
\end{equation}
\hfill \QED 
\end{theorem}

The solution (\ref{sol-K}) corresponds to partitioning the channel state space 
into {\em decision regions} $\Rc_{k,\ell} = \{ \sv \in \RR_+^K : w_{k,\ell}(\sv) = 1 \}$. 
These regions for $K > 1$ sensors are not just simple hyper-rectangles, and therefore
it is difficult to obtain a more explicit characterization. We will see in Section \ref{one-user} that for the single-sensor 
case the regions  are intervals that can be characterized more explicitly.  In general, 
letting $\Wm_{\omegav} (\sv)$ denote the  solution 
(\ref{sol-K}), we can define the  corresponding average sum energy:
\begin{equation} \label{energy-omega} 
\Gamma(\omegav) = \EE \left [ \bv^\transp(\Sm) \Wm_{\omegav} (\Sm) \av \right ], 
\end{equation}
and average rate vector
\begin{equation} \label{rate-omega} 
\Cm (\omegav) = \EE \left [ \Wm_{\omegav} (\Sm) \rv \right ]. 
\end{equation}
Although the decision regions defined by Theorem \ref{multi-sensor-min-power} do not lead, in general, to a closed-form
for the expectations in (\ref{energy-omega}) and in (\ref{rate-omega}), these can be easily computed by Monte Carlo 
averaging over the channel fading state $\Sm \sim F_S(\sv)$. 
Eventually, the minimum energy function for given threshold SNRs vector $\av$ and rate vector $\rv$, 
can be obtained by solving the Lagrangian dual problem with respect to $\omegav$, i.e., by maximizing
$\Lc(\Wm_{\omegav}, \omegav)$ over $\omegav \in \RR_+^K$.\footnote{Since the constraint corresponding to the dual variable $\nu(\sv)$ must be 
satisfied with equality for the solution $\Wm_{\omegav}$ in (\ref{sol-K}), then $\nu(\sv)$ 
becomes irrelevant and it is dropped for notational simplicity.}  
This can be obtained by a subgradient method, starting from some initial value $\omegav(0) \in \RR_+^K$, 
and successively updating the dual variables according to
\begin{equation} \label{sub-iteration} 
\omegav(n+1) = \omegav(n) + \epsilon_n \vv(n),
\end{equation}
where $\vv(n)$ is a subgradient for the problem.  A subgradient can be found by noticing that, for any $\omegav', \omegav \in \RR_+^K$, 
since $\Wm_{\omegav'}$ is a minimizer of $\Lc(\Wm,\omegav')$,  we have
\begin{eqnarray} \label{subgradient}
\Lc(\Wm_{\omegav'}, \omegav') 
& = & \Gamma(\omegav') - (\omegav')^\transp (\Cm(\omegav') - \lambdav ) \nonumber \\
& \leq & \Gamma(\omegav) - (\omegav')^\transp (\Cm(\omegav) - \lambdav ) \nonumber \\
& \leq & \Lc(\Wm_{\omegav}, \omegav)  + (\omegav' - \omegav)^\transp (\lambdav - \Cm(\omegav) ).
\end{eqnarray}
It follows that a subgradient at $\omegav(n)$ is given by the vector $\vv(n) = \lambdav - \Cm(\omegav)$. 
The sub-gradient iteration (\ref{sub-iteration}) with this choice of the subgradient has an intuitive meaning:
if for some $k$, the $k$-th component of the vector $\lambdav - \Cm(\omegav)$ is positive 
(i.e., the average service rate of sensor $k$ is below its arrival rate $\lambda_k$), then the corresponding weight $\omega_k(n)$ is increased.  
Otherwise, it is decreased. The iteration step can be chosen as $\epsilon_n = \epsilon_0 \frac{1+b}{n+b}$, for suitable parameters 
$\epsilon_0, b > 0$.  The resulting sequence of values $\Gamma(\omegav(n))$ obtained by the above iteration 
converges to $\Phi(\lambdav)$ as $n \rightarrow \infty$ (in practice, we checked that convergence is very fast). 

\subsection{Achievable upper bound on the minimum energy function} 

Driven by the intuition developed for the case where the PHY mode successful decoding probabilities are indicator functions, we 
fix the receiver SNR threshold values $\av = (a_0, \ldots, a_L)^\transp$, with $a_0 = \SNR_0$, and define the 
deterministic energy allocation function $e_{k,\ell}(\sv) = \frac{a_\ell}{s_k}$. Replacing this into problem (\ref{K-user-Phi-smooth1}), we obtain  that the optimization with respect to $\Wm(\sv)$ is again given in the form (\ref{K-user-Phi1}), 
where now $\av$ is fixed {\em a priori}, and $\rv = (0, R_1P_1(a_1), \ldots, R_LP_L(a_L))^\transp$. 
In the example of Fig.~\ref{PERfig-example}, $\mu(\SNR)$ is shown 
together with the piecewise constant curve defined by the points $\{ (a_\ell, R_\ell P_\ell(a_\ell) : \ell = 0, \ldots, L\}$, 
corresponding to a specific choice of the target received SNRs $a_\ell$. 
Eventually, we arrive at a linear program formally identical to what we have already solved, 
for a given set of receiver SNR values and for the corresponding set of {\em effective rates}
$r_\ell = R_\ell P_\ell(a_\ell)$. Therefore, Theorem \ref{multi-sensor-min-power} and the subgradient iteration can be applied verbatim,  
yielding the solution $\Wm(\sv;\av)$ and an achievable upper bound $\Phi(\lambdav; \av) \geq \Phi(\lambdav)$.  

The achievable upper bound can be tightened by optimizing over $\av$, through an educated exhaustive search over an appropriate domain.
In order to determine a suitable search domain, consider the typical behavior of 
$\mu(\SNR)$ as shown in the example of Fig.~\ref{PERfig-example}.
We notice that for each $\ell$-th PHY mode there exists a narrow interval where the effective 
rate presents a sharp  transition between very small to almost $R_\ell$ effective rate. 
This corresponds to the ``waterfall'' of the successful decoding probability $P_\ell(\SNR)$ (see Fig.~\ref{PERfig}). 
The optimal values of $a_\ell$ is found in this transition interval, indicated by $[\SNR^{(0)}_\ell, \SNR^{(1)}_\ell]$. 
As a rule of thumb, we choose $\SNR^{(1)}_\ell$ such that $P_\ell(\SNR^{(1)}_\ell) = 0.99$, and 
$\SNR^{(0)}_\ell$ such that $P_\ell(\SNR^{(0)}_\ell) = 0.1$. 
We argue that searching outside this interval is useless.  In fact, for $a_\ell < \SNR^{(0)}_\ell$ the contribution of 
PHY mode $\ell$ to the overall average rate is too small,  and therefore this mode is never selected by the scheduling policy (i.e., the corresponding probability $w_{k,\ell}(\sv)$ solution of  (\ref{K-user-Phi1}) is zero for all $k$ and $\sv$).
In contrast, for $a_\ell > \SNR^{(1)}_\ell$ the contribution of PHY mode $\ell$ to the average rate
does not increase, since it is essentially already almost equal to its upper bound $R_\ell$, 
while the contribution to the average energy increases linearly with $a_\ell$. 
Therefore, the proposed achievable upper bound $\overline{\Phi}(\lambdav)$ to the minimum energy function 
is obtained by  minimizing $\Phi(\lambdav;\av)$ over $\av$ in the Cartesian product 
region  $\prod_{\ell=1}^L \left [ \SNR^{(0)}_\ell , \SNR^{(1)}_\ell \right ]$, for  such appropriately defined intervals.
For small $L$, as the Bluetooth-like system  at hand, this can be done by discretizing this $L$-dimensional hyper-rectangular search 
region  and exhaustively calculating $\Phi(\lambdav; \av)$ 
for each $\av$ in the  discretized grid. Notice that this complexity depends on the number of transmission modes $L$ and 
on the shape of the mode effective rate function (e.g., see Fig.~\ref{PERfig-example}), and not on the number of 
sensors $K$.  We wish to remark that this search is needed only for the purpose of performance evaluation, 
and it is performed off-line.  This has {\em no impact on the complexity of the dynamic scheduling algorithms} 
of Sections \ref{opt-scheduling} and \ref{switch}, as we shall discuss later. 

\subsection{Lower bound on the minimum energy function} 

As mentioned before, we construct a lower bound for $\Phi(\lambdav)$ by finding a piecewise linear concave
function $\widetilde{\mu}(\SNR) \geq \overline{\mu}(\SNR)$ and use this in the relaxed problem (\ref{K-user-Phi-smooth2}). 
The rate function upper bound is obtained as follows.  Let $\widetilde{a}_0 = \SNR_0$ and $\widetilde{\mu}_0 = 0$. 
Then, for $\ell = 0,1,\ldots, L-1$, find  the straight line passing through the point $(\widetilde{a}_\ell, \widetilde{\mu}_\ell)$, strictly upperbounding 
$\mu(\SNR)$ for all $\SNR > \widetilde{a}_\ell$ with the exception of at most one tangent point. 
Let the straight line be given by 
\[ y = m_\ell (x - \widetilde{a}_\ell) + \widetilde{\mu}_\ell, \]
for some $m_\ell \geq 0$,  and find the intercept of  this line with the horizontal line $y = R_{\ell+1}$.  
Denote the abscissa of this intercept by $\widetilde{a}_{\ell+1}$, let $\widetilde{\mu}_{\ell+1} = R_{\ell+1}$, 
let $\ell \rightarrow \ell + 1$ and repeat the procedure. 
In this way, by linearly interpolating the obtained points, we have constructed a piecewise linear 
function $\widetilde{\mu}(\SNR)$ with $L+1$ segments joining the points 
$\{ (\widetilde{a}_\ell, \widetilde{\mu}_\ell): \ell = 0, \ldots, L\}$, completed with 
a last horizontal segment $\widetilde{\SNR} = R_L$ for $\SNR \geq \widetilde{a}_L$. 
The example of Fig.~\ref{PERfig-example}, shows $\widetilde{\mu}(\SNR)$ obtained as said above, 
together with the actual non-concave rate function $\mu(\SNR)$. 

From the proof of Lemma \ref{indicators} it is apparent that the relaxed problem (\ref{K-user-Phi-smooth2}) obtained by 
replacing $\overline{\mu}(\SNR)$ with the concave piecewise linear upper bound 
$\widetilde{\mu}(\SNR)$ is equivalent to the original problem 
(\ref{K-user-Phi-smooth1}), with indicator function probabilities  $P_\ell(\SNR) = 1\{ \SNR \geq \widetilde{a}_\ell\}$. 
Hence, Theorem \ref{multi-sensor-min-power} and the subgradient 
search illustrated before can be used to efficiently obtain the sought lower 
bound $\underline{\Phi}(\lambdav)$. 

\subsection{A numerical example}

As a concluding example, we show in Fig.~\ref{bounds} the upper and lower bounds to $\Phi(\lambda)$ for the case of 
a single sensor $K=1$,  and the rate functions of Fig.~\ref{PERfig-example}, obtained with the fading 
cdf $F_S(\cdot)$ used throughout all numerical results in this paper, described below. 
Three curves are shown.  The solid line corresponds to $\overline{\Phi}(\lambda)$ obtained as explained before, including the optimization over
the receiver SNR vector $\av$. The dotted line shows a non-optimized upper bound, obtained by choosing
the threshold SNRs such that the corresponding $P_\ell(a_\ell) = 0.99$. These are the same values corresponding to the 
piecewise constant rate function of Fig.~\ref{PERfig-example}, and are included here to show the effect of searching over
$\av$. Finally, the dashed line corresponds to $\underline{\Phi}(\lambda)$ calculated as explained above, using 
the function $\widetilde{\mu}(\SNR)$ shown in Fig.~\ref{PERfig-example}.

\paragraph{Remark on the channel state statistics}
The problems leading to $\Phi(\lambdav)$, $\underline{\Phi}(\lambdav)$ and $\overline{\Phi}(\lambdav)$
are generally feasible if the products  $a_\ell \EE[1/S_k | \Sm \in \Rc_{k,\ell} ]$ are finite for all $k,\ell$. 
This may not hold for some  fading distributions.  For example, in the case of Rayleigh (resp., Ricean) fading, 
$S_k$ is central (resp., non-central) chi-squared with two degrees of freedom. 
For $K = 1$ (single sensor)  and $\Rc_{1,\ell}$ containing the origin, 
then the expectation of the inverse channel fading is unbounded. In particular, for the case 
$K = 1$  we have the policy decision region (see Section \ref{one-user}) $\Rc_{1,0} = [0, s_1)$ for 
some fading threshold $s_1 > 0$, so that $\SNR_0 \EE[1/S | S \in [0, s_1)]$ 
is unbounded,  unless  $\SNR_0 = 0$.  It should be noticed here that for any $K > 1$ the {\em multiuser diversity} inherent in the system 
is sufficient to achieve finite average energy even for $\SNR_0 > 0$.  Nevertheless, in order to avoid analytical problems with the case 
$K = 1$, $\SNR_0 > 0$,  in all numerical results presented in this paper we used a truncated Ricean fading 
distribution \cite{proakis-book} with probability density function (pdf)  $f_S(s)$ with support 
$s \in [s_{\min}, \infty)$, with Rice factor $\Kc = 6.95$ dB,
unit second moment and where $s_{\min} = 0.01$ corresponding to a deep fade 
of $-20$ dB. These fading statistics are relevant for typical WBAN.

\paragraph{Remark on the Bluetooth packet length}
In the Bluetooth-like system model, each transmission can span 1, 3 or 5 time slots with different rate/successful detection 
probability tradeoffs.  
We assume that the span of each packet is not adapted on a per-slot basis, 
but it is optimized offline, based on the fading statistics. 
Therefore,  the number of time slots per payload is fixed for all transmissions. 
We evaluated the $\overline{\Phi}(\lambdav)$  for each of the packet lengths using the method given in this section, 
and we found that for a wide range of arrival rates  the minimum energy is  achieved for payload length of 3 time slots. 
Hence, in all numerical results presented in this work we set the number of slots per payload to 3, i.e., 
we used PHY  modes 2dh3 and 3dh3 of Table \ref{PacketSlots} and Fig. \ref{PERfig}. 
The corresponding effective rates, including the protocol overhead, are given by $R_1 = 2 \times 367/371$ and 
$R_2 =  3 \times 552/556$.

\subsection{A closer look at the single-sensor case} \label{one-user}

In this section, we take a closer look at the single-sensor case and provide 
more insight on the structure of the policy achieving $\Phi(\lambda)$ (in the case 
$P_\ell(\SNR) = 1\{ \SNR \geq a_\ell\}$)  or $\overline{\Phi}(\lambda)$ (in the general smooth $P_\ell(\SNR)$ case). 
In both this section and Section \ref{switch}, considering a single sensor, 
we drop the sensor index for the sake of notation simplicity.

The decision regions of Theorem \ref{multi-sensor-min-power}, in this case, are given by:
\begin{equation} \label{dec-region1}
\Rc_\ell = \left \{ s \in \RR_+ \; : \; \omega r_\ell  - \frac{a_\ell}{s}  \geq \omega r_j  - \frac{a_j}{s}, \; \forall \; j \neq \ell \right \} 
\end{equation}
In order to proceed further,  we assume that (as in Lemma \ref{indicators}) the points $(a_0, 0)$, $(a_1,r_1)$, 
$(a_2,r_2)$,  $\ldots$,  $(a_L,r_L)$ define a concave non-decreasing piecewise linear function 
for $\SNR \geq a_0$. 
This corresponds to the general principle of diminishing return that typically occurs in communications 
channels, such that any additional increment of received SNR provides proportionally less and less rate as SNR increases. 
For example, the SNR thresholds $a_\ell = \SNR_\ell = 2^{r_\ell} - 1$ corresponding to
a family of capacity-achieving codes for the circularly symmetric AWGN channel satisfies this property.

Under this concavity condition, it is immediate to see that the regions $\Rc_\ell$ are in fact intervals 
that partition the channel state space $\RR_+$, for any value of the Lagrange multiplier $\omega > 0$. 
Defining the breakpoints $s_0 = 0$, $s_{L+1} = +\infty$ and
\begin{equation} \label{state-thresholds} 
s_\ell = \frac{a_\ell - a_{\ell-1}}{\omega (r_\ell  - r_{\ell-1})}, 
\end{equation}
for $\ell = 1,2,\ldots, L$. The concavity condition implies the ordering  $s_0 < s_1 \leq s_2 \leq \cdots \leq s_L < s_{L+1}$ and, 
consequently,  $\Rc_\ell = [s_\ell, s_{\ell+1} )$.  

It follows that the optimization solution in the single sensor case  takes on the intuitive form of a 
mode selection  strategy that depends on the strength of the channel state: if $S
\in [s_\ell, s_{\ell+1})$,  then mode $\ell$ is chosen with transmit energy per symbol $\frac{a_\ell}{s}$. 
As already observed for Lemma \ref{indicators},  if the concavity condition is not satisfied, 
some PHY modes are irrelevant for the minimum energy solution. 
We can re-define the problem for $L' < L$ modes, satisfying the concavity condition. 
Notice the upper concave envelope $\overline{\mu}(\SNR)$ of the SNR/rate points
$\{(a_\ell, r_\ell)\}$ yields the achievable long-term average rate versus received SNR curve obtained 
by time-sharing between the different PHY modes. If a point lies below this curve, it means that a better SNR/rate 
point can be achieved by time-sharing  between two other such points. Intuitively, this explains why these points can be 
discarded from the optimization. 

By using the explicit expression of the decision regions into the expression of the average energy and 
rate (\ref{energy-omega}) and (\ref{rate-omega}), we obtain
\begin{equation} \label{energy-omega1} 
\Gamma(\omega) = \sum_{\ell=0}^L \int_{s_\ell}^{s_{\ell+1}} \frac{a_\ell}{s} dF_S(s),    
\end{equation}
and
\begin{equation} \label{rate-omega1} 
C (\omega) =  \sum_{\ell=1}^L r_\ell \int_{s_\ell}^{s_{\ell+1}} dF_S(s),  
\end{equation}
where the dependence on $\omega$ is contained in the thresholds $\{s_\ell\}$ through (\ref{state-thresholds}).
Finally, the Lagrange multiplier is obtained by solving (numerically) the equation
$C(\omega) = \lambda$. 

When we evaluate the achievable upper bound $\overline{\Phi}(\lambda)$, 
searching for the optimal target receiver SNR vector $\av$ is easily done by using the almost closed-form expressions
(\ref{energy-omega1}) -- (\ref{rate-omega1}).  Then, the same values (which depend only on the protocol modes and not on the number of sensors), 
can be reused for the multisensor case (recall that any choice of $\av$ yields a valid upper bound). 
Fig.~\ref{sec5} shows $\Phi(\lambda)$ for $P_\ell(\SNR)$ given by  Fig. \ref{PERfig} for the 
two enhanced data rates 2dh3 and 3dh3, a single sensor ($K = 1$), and different choices of the NULL packet target SNR.  
For comparison, we also show $\Phi(\lambda)$ for a system with the same PHY mode rates of 2 and 3 bit/s/Hz, but exploiting
ideal capacity achieving codes and $\SNR_0 = 0$. The very large gap between the minimum energy function of Bluetooth-like systems
and the ideal system stresses the fact that the PHY layer of Bluetooth is very suboptimal, in exchange for very low complexity.

\section{ÒMinimum average energy subject to (finite) average delay constraints}  \label{opt-scheduling}

In this section we consider the  more general problem of
minimizing the transmit energy subject to an average (finite) delay constraint, 
particularize the scheduling policy developed in \cite{Neely-IT07}  to our setting, in order to
operate closely to any desired point of the optimal energy-delay tradeoff curve. 
This policy depends on the control parameter $V \in \RR_+$, and provably achieves near-optimal 
energy-delay tradeoff in the following sense: for $V > R_L$, the policy achieves delay $\overline{D} = O(\sqrt{V} \log V)$ with average 
transmitted sum-energy per symbol $\overline{E} - \Phi(\lambdav) = O(1/V)$. 
For twice-differentiable $\Phi(\lambdav)$, it is known that 
any policy must satisfy the Berry-Gallager bound: if $\overline{E} - \Phi(\lambdav) = O(1/V)$ then 
$\overline{D} = \Omega(\sqrt{V})$ \cite{Neely-IT07}.  In our case, since the channel fading has a continuous distribution, 
$\Phi(\lambdav)$  is a strictly convex continuous and increasing function of $\lambdav$ 
and therefore it is twice differentiable almost everywhere \cite{rockafellar1996convex} with strictly positive Hessian.

The maximum absolute variation of $Q_k(t)$ for any  $k$ from one slot to the next is bounded by 
$\delta_{\max} = \max\{ A_{\max}, R_L\}$, where $A_{\max}$ is the bound on the arrival process. 
We assume a strict peak power constraint $\Pc_{\max}$, such that $\Pc_{\max} s_{\min} \geq \SNR_0$, therefore, 
the minimum required received SNR while transmitting rate $R_0 = 0$ (NULL control packets) can always be achieved. 
Defining the constants $Q_{th}, \zeta, \nu > 0$ such that, for $V > 0$, 
$\zeta = \frac{\nu}{\delta_{\max}^2} e^{-\nu/\delta_{\max}}$, 
$\nu = \frac{1}{\sqrt{V}}$ and $Q_{th} = \frac{6}{\zeta} \log \frac{1}{\nu}$ (see \cite{Neely-IT07}) the dynamic scheduling policy proceeds as follows. 

The scheduling policy is given in Algorithm \ref{algo1}. 
At each slot-time $t$, the scheduler chooses a sensor $k(t)$, 
a transmit energy per symbol  $E_{k(t)}(t)$  and  a mode $\ell(t)$.  
This selection is given as the result of the optimization step (\ref{Eopt}), performed at each slot time $t$.
The weights in this optimization are computed by updating {\em auxiliary queues}
$\Xm(t) = (X_1(t), \ldots, X_K(t))$ according to
\begin{eqnarray} \label{X}
X_k(t + 1) & = & \left [ X_k(t) - \mu(E_k(t) S_k(t)) -  \nu 1\{Q_k(t) < Q_{th} \} \right ]_+ \nonumber \\
& & + A_k(t) +  \nu 1\{Q_k(t) \geq Q_{th} \}
\end{eqnarray}
where $\mu(\SNR)$ is defined in (\ref{mu}), extended such that $\mu(\SNR) = 0$ for $\SNR \leq \SNR_0$.  
Also, it is understood that if $k(t) \neq k$ then $E_k(t) = 0$.  

\begin{algorithm}
\caption{Multi-sensor Opportunistic Scheduler.}
\label{algo1}
\begin{itemize}
\item Initialize $\Xm(0)= \zerov$.
\item {\bf for} $t = 1,2,...$ {\bf repeat :}
\begin{itemize} {\bf for} $k=1, \ldots, K$ {\bf let}: 
\begin{eqnarray}
\widehat{E}_k &=& \arg\min_{e \in [\SNR_0/S_k(t), \Pc_{\max}]} \left \{ V e -  \left [ W_k(t) \right ]_+ \;  \mu(e S_k(t)) \right \} 
\label{Eopt}
\end{eqnarray}
where
\begin{eqnarray} \label{W}
W_k(t) & = &  1\{Q_k(t) \geq Q_{th} \} \times \zeta e^{\zeta(Q_k(t) - Q_{th})} \nonumber \\
& & - 1\{Q_k(t) < Q_{th} \}  \times \zeta e^{-\zeta(Q_k(t) - Q_{th})} + 2X_k(t)
\end{eqnarray}
and let 
\[ \widehat{\ell}_k = \left \{ \begin{array}{ll}
0 & \mbox{if} \;\;\; \widehat{E}_k S_k(t) = \SNR_0 \\
\arg\max_{\ell = 1,\ldots, L} \{ R_\ell P_\ell(\widehat{E}_k S_k(t)) \} & \mbox{if} \;\;\; \widehat{E}_k S_k(t) > \SNR_0 
\end{array} \right . \]
{\bf end for}
\end{itemize}
\item Let $k(t)  = \arg\min_{k=1,\ldots,K} \left \{  V \widehat{E}_k -  \left [ W_k(t) \right ]_+ \;  \mu(\widehat{E}_k S_k(t)) \right \}$, 
let $\ell(t) = \widehat{\ell}_{k(t)}$, $E_{k(t)}(t) = \widehat{E}_{k(t)}$  and $E_k(t) = 0$ for all $k \neq k(t)$. 
\item Let sensor $k(t)$ transmit with energy $E_{k(t)}(t)$ and PHY mode $\ell(t)$ on slot $t$. 
\item For the current arrival vector $\Am(t)$ update the queue buffers and scheduler 
weights according to  (\ref{buffer-evo}) and (\ref{X}), respectively. \\
{\bf end for}
\end{itemize}
\end{algorithm}

The dynamic policy of Algorithm 1 requires the on-line computation of the minimum in (\ref{Eopt}) for all $k = 1,\ldots, K$.
For $W_k(t) \leq 0$, (\ref{Eopt})  yields $\widehat{E}_k = \SNR_0/S_k(t)$. 
For $W_k(t) > 0$, using (\ref{mu}), the minimization in  (\ref{Eopt}) is equivalent to calculating
\begin{equation} \label{ziocanebastardo}
\min_{e \in [\SNR_0/S_k(t), \Pc_{\max}]} \; \left \{ V e - \left [ W_k(t) \right ]_+ \;  R_\ell P_\ell(e S_k(t)) \right \}
\end{equation}
for $\ell = 0, \ldots, L$, and choosing the energy value that achieves the overall minimum. 
This can be obtained by performing $L$ one-dimensional line searches~\footnote{For $\ell = 0$ the minimum
is achieved at $e = \SNR_0/S_k(t)$ therefore no line search is needed.}
over the interval $[\SNR_0/S_k(t), \Pc_{\max}]$, with linear complexity in $L$. 
Also, we noticed that by computing (\ref{ziocanebastardo}) at the energy values
$e_\ell = a_\ell / S_k(t)$, where $\{a_\ell\}$ are the SNR thresholds resulting from the calculation 
of $\overline{\Phi}(\lambdav)$, and selecting the PHY mode accordingly, 
the performance of Algorithm 1 is practically indistinguishable from the case where the full line search is performed. 
In fact, if the probabilities $P_\ell(\SNR)$ are indicator functions, then the the minimum of (\ref{ziocanebastardo}) 
must be one of the break-points $\{\SNR_\ell/S_k(t)\}$. Hence, given the shape of the probabilities $P_\ell(\SNR)$ (see Fig.~\ref{PERfig-example}),
it is not surprising that the actual minimum of (\ref{ziocanebastardo}) is very close to the value at
$e_\ell = a_\ell / S_k(t)$.  We conclude that the complexity of the on-line optimization performed 
in Algorithm 1 is $O(L)$  with respect to the number of PHY modes,  and $O(K\log K)$ with respect to the number 
of sensors, due to the presence of the maximization  over $k$.  In practical WBANs, both $L$ and $K$ are small integers, 
therefore these complexity orders are quite irrelevant. As a matter of fact,  Algorithm 1 can be easily implemented in the piconet hub 
(typically a powerful smartphone, equipped with a modern multi-core processor). 

In order to gain insight on the gain achievable by the multi-sensor dynamic scheduling policy given in this section, we compare 
its energy-delay performance with that of a suboptimal system that runs a single-user energy-delay algorithm independently 
for each sensor  (see Section \ref{one-user}), and uses round-robin channel allocation to given to all sensors the same fraction 
of system bandwidth. We considered a piconet with two sensors with same channel state statistics and two setups:
1) symmetric arrival rates $\lambda_1 = \lambda_2$, and deterministic packet arrivals (one packet per slot per sensor); 
2) symmetric arrival rates with different probabilities of packet arrival $q_1=1.0$ and $q_2=0.2$. In setup 2, 
the arrival process of sensor 2 is i.i.d. with geometrically distributed inter-arrival time with
$\PP(A_2(t) = 0) = 1 - q_2$ and $\PP(A_2(t) = \lambda_2/q_2) = q_2$.
The corresponding energy-delay tradeoff curves are shown in Fig. \ref{opportunistic}.
The achievable sum average energy $\overline{\Phi}(1,1)$ yields the horizontal line, and is computed using 
method described in Section \ref{min-pow-K}, for the packet success probabilities $P_\ell(\SNR)$ of 
Bluetooth 2dh3 and 3dh3 PHY modes,  with $\SNR_0 = 8$ dB. 
We notice that the multisensor scheduling policy, which selects opportunistically the sensor to serve on each slot, 
based on the channel state and on the queue state, achieves a significant gain in average transmit energy with respect to a
conventional round-robin piconet scheduling. 
By increasing the control parameter $V$ the average energy tends to
its minimum, while the delay increases. 
We also considered an asymmetric setup with deterministic arrivals (one packet per slot) but 
the two arrival rates are different, namely $\lambda_1 = 0.04$ bit/s/Hz and $\lambda_2 = 1.0$ bit/s/Hz. 
The corresponding energy-delay tradeoff curves are shown in Fig. \ref{opportunistic2}. 
These examples indicate that the energy savings of the multi-sensor dynamic scheduling policy
increase when the sensors are asymmetric, either because of the different inter-arrival 
statistics, and because of the different arrival average rate.

\section{Improved policy with adaptive sleep mode optimization}  \label{switch}

The scheduling policy of Section \ref{opt-scheduling} is near-optimal under the assumption that there is no power expenditure when 
no payload data is transmitted  \cite{Neely-IT07,Berry-Gallager-IT02}. However, as already discussed a few times in this paper, in Bluetooth or similar systems 
a transmitter must  spend power for transmitting NULL control packets even when it has no data to transmit. 
In this section we investigate an adaptive policy that switches 
between sleep and connected modes. We focus first on the single-sensor case, and then discuss the extension to the 
case of $K > 1$. 

By introducing a sleep mode, the transmitter has two distinct states and the scheduling problem 
becomes significantly more complicated.  
A rather general formulation of the problem was addressed in terms of dynamic programming 
over ``renewal intervals'' in \cite{neely2009stochastic}, where 
the event of disconnecting the sensor and moving to sleep mode is a recurrent event that somehow 
``resets'' the system. For the sake of a simple on-line implementation, here we seek a heuristic approach that yields 
significant improvements with respect to the dynamic policy of Algorithm \ref{algo1}.

In order to motivate our approach, we observe the qualitative behavior of the policy of 
Algorithm \ref{algo1} for fixed arrival rate $\lambda$ with different levels of ``burstiness'' of the arrival process. 
Consider Fig. \ref{kazzi}, obtained by Algorithm \ref{algo1} in the case of a bursty arrival 
process $A(t)$ with i.i.d. arrivals, 
arrival rate $\lambda = 0.04$ bit/s/Hz and $\PP(A(t) = \lambda/q) = q$ and $\PP(A(t) = 0) = 1 - q$. 
We notice that the scheduler tends to allocate high power and high rate immediately 
after a non-zero data arrival,  when the queue buffer is large. 
Then, as $Q(t)$ becomes smaller and eventually drops slightly the threshold $Q_{th}$ (see Algorithm \ref{algo1}), the scheduler 
tends to allocate smaller rates and smaller transmit powers. 
Eventually, it uses the minimum power necessary for NULL packets with zero 
payload  data rate, until the next arrival. 

From (\ref{buffer-evo}), (\ref{X}) and (\ref{W}) we have that if  $Q(t) < Q_{th}$ and no arrivals occur at slot $t$, then necessarily  
$X(t + 1) < X(t)$  (i.e., the auxiliary queue has a locally decreasing trajectory).  Furthermore,  when $X(t)$ is sufficiently 
small such that $W(t) \leq 0$, the coefficient $W(t)$ is certainly non-positive until the next arrival.  
It follows that when the algorithm enters the condition $W(t) = 0$, it keeps scheduling NULL slots until the 
next arrival time. Let $t_0$ and $t_1$ denote two consecutive arrival times, and assume that the condition $W(t) = 0$ 
occurs at $t_w$, i.e.,  we define the conditional stopping time
\[ t_w = \min_{t > t_0} \;\; \{ W(t) \leq 0 : \mbox{arrival at time} \; t_0 \}. \]
Assuming that at time $t_1$ the buffer queue jumps above the threshold $Q_{th}$ and the algorithm 
schedules non-NULL slots,  the average energy per symbol spent in the ``idle'' phase is given by 
\begin{eqnarray} \label{conditional-idle-power} 
\overline{E}_{\rm idle}(t_w  ,t_0) 
& = & \EE \left [ \left . \sum_{\tau = t_w}^{t_1} \frac{\SNR_0}{S(\tau)}   \right | t_1 > t_w , t_0\right ] \nonumber \\
& = & \EE\left [ \frac{\SNR_0}{S} \right ]   \EE\left [ \left . t_1 - t_w  + 1 \right | t_1 > t_w,  t_0 \right ] \nonumber \\
& = & \overline{E}_{\rm NULL} \times \Delta(t_w, t_0), 
\end{eqnarray}
where we define $\overline{E}_{\rm NULL} =  \EE\left [ \frac{\SNR_0}{S} \right ]$ as the average energy necessary to 
transmit a NULL packet, and the term $\Delta(t_w, t_0) = \EE\left [ \left . t_1 - t_w  + 1 \right | t_1 > t_w,  t_0 \right ]$ is the conditional 
average idle time  for given $t_w$ and $t_0$.  For example, in the case of deterministic arrivals, which are representative of sensors that collect 
one measurement  of $A_{\max}$ bit/s/Hz each interval of (say) $N$ slots, we have  
\begin{eqnarray} \label{idle-time-determ}
\Delta(t_w, t_0) & = & \EE\left [ \left . t_1 - t_w  + 1 \right | t_1 > t_w,  t_0 \right ] = N - (t_w  - t_0) + 1. 
\end{eqnarray}
In the case of geometrically distributed inter-arrival times with $\PP(A(t) = A_{\max}) = q$,  we have 
\begin{eqnarray} \label{idle-time-geom}
\Delta(t_w, t_0) & = & \EE\left [ \left . t_1 - t_w  + 1 \right | t_1 > t_w,  t_0 \right ] \nonumber \\
& = & \EE\left [ \left . t_1 \right | t_1 > t_w \right ] - t_w + 1 = \frac{1}{q} + 1.
\end{eqnarray}
Assuming that $\Delta(t_w, t_0)$ is known a priori from the system statistics, we propose the following
switching strategy in order to take advantage of the sleeping mode: at each new arrival $t_0$,  
the policy resets.  Focusing on a given inter-arrival period and letting $t_0$ be the starting time (i.e., the time of the most recent arrival), 
if $t_0 < t_w < t_1$ occurs, i.e., the scheduler enters the idle interval before a new arrival occurs, 
then the policy switches to sleep mode if $\overline{E}_{\rm idle}(t_w,t_0) > \overline{E}_{\rm sleep}$, 
where  the latter is the average energy cost of disconnecting and setting up a new connection as soon as a new arrival occurs.
In general,  $\overline{E}_{\rm sleep}$ depends on the paging and hand-shaking procedures, which 
involve the transmission of some control packets. In our results, we let $\overline{E}_{\rm sleep} = \tau \overline{E}_{\rm NULL}$ 
for some  value $\tau > 1$, since we assume that the energy cost of disconnecting and reconnecting at a later time
is some multiple of the NULL control packet average energy.   
Then, the improved policy switches to sleeping mode if  $\Delta(t_w, t_0) > \tau$. 
Fig. \ref{switch-fig} compares the energy-delay tradeoff obtained by the proposed switching policy for different values of 
$\tau$, for the same arrival statistics of the snapshots of Fig.~\ref{kazzi}. The case $\tau = 10$ corresponds to 
no switching because the cost of disconnecting and reconnecting is equal to the expected cost of NULL packets in idle times.
In this case, the curve coincides with what obtained by Algorithm \ref{algo1}. 

In order to extend this idea to the case of $K > 1$ sensors, we observe that when there exists at least one sensor 
with data to transmit in its queue buffer, it is very likely that the hub will poll this sensor. 
In this case, even though some sensors are idle,  they are unlikely to be polled and therefore they need not send 
NULL packets. Instead, if all sensors have below-threshold buffers, the minimization in (\ref{Eopt}) yields 
$\widehat{E}_k = \SNR_0/S_k(t)$ for all $k$.  
Hence, the sensor with the largest value of channel state $S_k(t)$ will be polled and has to respond 
with a NULL packet.  For independent identically distributed stationary and ergodic fading processes across the sensors, 
each sensor is equally likely to have the largest fading state, therefore, all sensors have the same probability $1/K$ 
of sending a NULL packet. 
An example of this situation is evidenced by the snapshot of Fig. \ref{kazzi2}, obtained by applying Algorithm \ref{algo1} in 
the case of $K = 2$ sensors,  with i.i.d. channel state and arrival processes, each of which has the same statistics 
as in Fig.~\ref{kazzi}.  Based on these observations, for $K > 1$ we define $t_0$ to be the time of the last data arrival over all sensors, 
and the stopping time $t_w$ as
\[ t_w = \min_{t > t_0} \;\; \{ W_k(t) \leq 0 \;\; \forall \; k : \mbox{last arrival at time} \; t_0 \}.  \]
In words, $t_w$ is the first time where all sensors become idle since the last packet arrival in the whole piconet.
Then, each sensor $k$ calculates an estimate of the expected energy spent to send NULL packets as
$\overline{E}_{{\rm idle},k}(t_w  ,t_0)  = \frac{\overline{E}_{\rm NULL}}{K} \times  \Delta_k(t_w, t_0)$, where 
$\Delta_k(t_w, t_0)$ is defined as in (\ref{idle-time-geom}) for the arrival process of sensor $k$, and
where the factor $1/K$ is due to the fact that, when all sensors are idle, 
they transmit NULL packets for $1/K$ of the slots, on average. 
Each sensor $k$ switches to sleep mode if $\overline{E}_{{\rm idle},k}(t_w  ,t_0) > \overline{E}_{\rm sleep}$. 
Suppose that at time $t_w$ some sensor switches to sleep mode. Then, at time $t_w+1$, if the condition 
$W_k(t_w+1) \leq 0$ still holds, then the same decision is repeated for the surviving network 
of $K' < K$ still connected sensors, by decrementing $\Delta_k(t_w, t_0)$ by one slot unit. The switching decision process 
repeats for $t_w+2, t_w + 3, \ldots$, until some new packet arrives into the network. 
Sleeping sensors turn ``on'' again when they receive a new non-zero data packet and their weights become positive again.

We expect that for $K > 1$ the proposed switching policy can achieve significant gains 
over the standard dynamic scheduling policy of Algorithm \ref{algo1} when all arrival processes 
are bursty. In contrast, if one or more sensors have full queue buffers 
most of the time, then the switching condition does not occur. In this case, sensors are rarely 
requested to send NULL packets, since there exists some sensor with data to send with 
high probability, and therefore the hub will likely poll these sensors. 

This intuition is confirmed by numerical simulations.  Fig.~\ref{switch-fig2} shows the energy-delay tradeoff of the proposed switching policy 
for $K = 2$ symmetric sensors with arrival statistics as in Fig.~\ref{kazzi2}. Fig.~\ref{switch-fig3} shows analogous curves for an asymmetric case
where $\lambda_1 = 0.04$ bit/s/Hz, $q_1 = 0.1$,  and $\lambda_2 = 1.0$ bit/s/Hz and $q_2 = 1.0$ (deterministic arrivals). 
The switching policy achieves significant gains over the non-switching policy in the case of Fig.~\ref{switch-fig2}, where both sensors
have bursty arrival processes. Instead, in the case of Fig.~\ref{switch-fig3}, sensor 2 has always data to transmit, and
the switching conditions never occurs. On the other hand, sensor 1 does not waste energy to send NULL packets 
during its own idle times, since sensor 2 is polled with high probability instead. 


\section{Conclusions} \label{conclusions}

We studied the minimum energy function and the energy-delay tradeoff in a wireless network 
inspired by  a Bluetooth piconet with $K$ sensors and a data-collection hub. 
This problem is motivated by body-area wireless sensor networks, where 
power efficiency plays a key role, but critical data must be delivered with low delay. 
This paper explicitly considers system features such as a finite (and typically small) number of possible 
data rates,  simple block codes (or even uncoded modulation) with non-negligible packet error probabilities, 
and the cost of protocol overhead packets (NULL packets) in order to maintain a connection in the absence of payload data. 
In this context, we applied a known general theory of near-optimal scheduling for energy-delay tradeoff
and have obtained an easily computable solution for the minimum energy function in the case where the codes packet error probabilities
are 0 if the receiver SNR is above some code-dependent threshold, and 1 elsewhere. This sharp transition behavior is relevant for very powerful codes
and large block length. In the case of moderate block length and/or of weak codes (such as in the Bluetooth-like system considered here), 
the packet error probabilities are smooth functions of the received SNR. In this case, we developed an easily computable achievable upper bound
and a lower bound to the minimum energy function.
Furthermore, we provided an explicit scheduling algorithm that provably operates along the optimal power-delay tradeoff, 
and at each slot performs sensor selection, power allocation and
transmission mode (and rate) selection. 
Finally, we proposed a novel improved strategy for the case where the transmission of NULL control packets has a 
non-zero power cost (this is in fact the most relevant and practical case). 
Our improved strategy adaptively switches the sensor to sleep mode, by comparing the conditional expected 
cost of maintaining the connection, given the present state, with the average cost of disconnecting and 
re-connecting later on. 

The proposed algorithms are amenable for an on-line implementation 
at the hub node, since in a Bluetooth-like piconet the hub node polls the sensors, and 
our strategy only requires a different ``opportunistic'' polling order, that depends on the sensors queue state and on the channel state. The problem of learning the channel fading gains at the transmitter to a sufficient level of accuracy goes beyond the scope of this work,  and should be addressed by some suitable two-way training scheme, exploiting time-division duplex operations, 
reciprocity and the rather large channel time-frequency correlation. 

\appendix

\section{Proofs of auxiliary results}  \label{proofs}

\subsection{Proof of Lemma \ref{deterministic-E}}

Using iterated expectation, we write the objective function in  (\ref{K-user-Phi-smooth2}) as
\[ \EE \left [ \sum_{k,\ell} w_{k,\ell}(\Sm) \Ec_{k, \ell}(\Sm)  \right ] = \EE \left [ \sum_{k,\ell} w_{k,\ell}(\Sm) \EE[ \Ec_{k, \ell}(\Sm) |\Sm] \right ]. \]
Using the concavity of $\overline{\mu}(\SNR)$, iterated expectation and Jensen's inequality in the terms appearing in the 
stability constraints we obtain
\begin{eqnarray}
\EE \left [ \sum_{\ell} w_{k,\ell}(\Sm) \overline{\mu} (\Ec_{k,\ell}(\Sm) S_k)   \right ] 
& = &  \EE \left [ \sum_{\ell} w_{k,\ell}(\Sm) \EE [ \overline{\mu} (\Ec_{k,\ell}(\Sm) S_k) | \Sm]   \right ] \nonumber \\
& \leq &  \EE \left [ \sum_{\ell} w_{k,\ell}(\Sm) \overline{\mu} (\EE [  \Ec_{k,\ell}(\Sm)| \Sm]  S_k)   \right ]. 
\end{eqnarray}
Hence, for any probability assignment $p_{\Ec_\pi |\Kc_\pi, \Lc_\pi, \Sm}( e | k,\ell,\sv)$ for $\Ec_{k,\ell}(\sv)$ we can 
find the feasible deterministic energy function $e_{k,\ell}(\sv) = \EE[\Ec_{k,\ell}(\Sm) | \Sm = \sv]$ that yields a smaller or equal 
value of the problem. 

\subsection{Proof of Lemma \ref{indicators}}

First, notice that under the assumption of Lemma \ref{indicators} the upper concave envelope $\overline{\mu}(\SNR)$
coincides with the piecewise linear function joining the points  $(a_\ell, R_\ell)$ for $\ell = 0, \ldots, L+1$. 
For any feasible $\Wm(\sv)$ and randomized energy allocation $\Ec_{k,\ell}(\sv)$ with conditional distribution
$p_{\Ec_\pi |\Kc_\pi, \Lc_\pi, \Sm}( e | k,\ell,\sv)$, define the quantized energy allocation
\begin{equation} \label{quantized-energy}
\Xi_{k,\ell}(\sv)  = \left \{ \begin{array}{ll} 
\frac{a_\ell}{s_k}, &  \mbox{for} \;\; \Ec_{k,\ell}(\sv) \geq \frac{a_\ell}{s_k} \\
\frac{a_0}{s_k}, & \mbox{for} \;\; \Ec_{k,\ell}(\sv) < \frac{a_\ell}{s_k}. \end{array} \right . 
\end{equation}
Notice that the equality $1\{ \Ec_{k,\ell}(\sv) s_k \geq a_\ell \} =  1\{ \Xi_{k,\ell}(\sv) s_k \geq a_\ell \}$ holds with probability 1 
for all $k, \ell$ and $\sv$.
Then, the terms in the stability constraints of (\ref{K-user-Phi-smooth1}) are preserved by 
replacing $\Ec_{k,\ell}(\sv)$ with $\Xi_{k,\ell}(\sv)$. Furthermore, since $\PP(\Xi_{k,\ell}(\sv) \leq \Ec_{k,\ell}(\sv)) = 1$ for 
all $k,\ell$ and $\sv$, we have that the objective function in (\ref{K-user-Phi-smooth1}) cannot increase
by replacing $\Ec_{k,\ell}(\sv)$ with $\Xi_{k,\ell}(\sv)$. Hence, the optimal 
energy allocation must be in the quantized form (\ref{quantized-energy}) with at most two mass points, one 
at $\frac{a_\ell}{s_k}$ and  the other at $\frac{a_0}{s_k}$, with probabilities indicated by 
$\xi_{k,\ell}(\sv)$ and $1 - \xi_{k,\ell}(\sv)$, 
respectively. 

When restricting 
the energy allocation to the two-mass point random variable (\ref{quantized-energy}), 
the objective function in (\ref{K-user-Phi-smooth1}) takes on the form
\begin{eqnarray} \label{obj-indicator}
\EE \left [ \sum_{k,\ell} w_{k,\ell}(\Sm) \Xi_{k, \ell}(\Sm) \right ] 
& = & 
\EE \left [ \sum_k \widetilde{w}_k(\Sm) \sum_{\ell=0}^L \frac{w_{k,\ell}(\Sm)}{\widetilde{w}_k(\Sm)} \left ( \frac{a_0}{S_k} (1 - \xi_{k,\ell}(\Sm)) + \frac{a_\ell}{S_k}
\xi_{k,\ell}(\Sm) \right ) \right ] \nonumber \\
& = & 
\EE \left [ \sum_k \widetilde{w}_k(\Sm) \left \{ \left ( \frac{w_{k,0}(\Sm)}{\widetilde{w}_k(\Sm)} + \sum_{\ell=1}^L \frac{w_{k,\ell}(\Sm)}{\widetilde{w}_k(\Sm)} (1 - \xi_{k,\ell}(\Sm)) \right ) \frac{a_0}{S_k} \right . \right . \nonumber \\
& & \left . \left . +  \sum_{\ell=1}^L \frac{w_{k,\ell}(\Sm)}{\widetilde{w}_k(\Sm)} \xi_{k,\ell}(\Sm)  \frac{a_\ell}{S_k} \xi_{k,\ell}(\Sm) \right \} \right ] \nonumber \\
& = & 
\EE \left [ \sum_k \widetilde{w}_k(\Sm) \sum_{\ell=0}^L \widetilde{w}_{\ell | k}(\Sm) \frac{a_\ell}{S_k}  \right ],  
\end{eqnarray}
where $\widetilde{w}_k(\sv) = \sum_{\ell=0}^L w_{k,\ell}(\sv)$ and
where we define the conditional probability mass function $\widetilde{w}_{\ell | k}(\sv)$ of choosing PHY mode $\ell$ given that
the sensor $k$ is scheduled, and the channel state is equal to $\sv$, given by 
\begin{eqnarray*}
\widetilde{w}_{0|k}(\sv) & = &  \frac{w_{k,0}(\sv)}{\widetilde{w}_k(\sv)} + \sum_{\ell=1}^L \frac{w_{k,\ell}(\sv)}{\widetilde{w}_k(\sv)} (1 - \xi_{k,\ell}(\sv)) \\
\widetilde{w}_{\ell | k}(\sv) & = & \frac{w_{k,\ell}(\Sm)}{\widetilde{w}_k(\Sm)} \xi_{k,\ell}(\Sm), \;\;\; \mbox{for} \;\; \ell > 0.
\end{eqnarray*}
The terms appearing in the stability constraints can be written as
\begin{eqnarray} \label{constr-indicator}
\EE \left [ \sum_{\ell} w_{k,\ell}(\Sm) R_\ell 1\{ \Xi_{k,\ell}(\Sm) S_k \geq a_\ell\}   \right ] 
& = & \EE \left [ \widetilde{w}_k(\Sm) \sum_{\ell = 1}^L  \frac{w_{k,\ell}(\Sm)}{\widetilde{w}_k(\Sm)} \xi_{k,\ell}(\Sm) R_\ell   \right ] \nonumber \\
& = & \EE \left [ \widetilde{w}_k(\Sm) \sum_{\ell = 0}^L  \widetilde{w}_{\ell | k}(\Sm) R_\ell   \right ].
\end{eqnarray}
At this point, we have reduced the solution of (\ref{K-user-Phi-smooth1}) to the determination of the conditional probability mass function $\widetilde{w}_k(\sv) \widetilde{w}_{\ell | k}(\sv)$ of $\Kc_\pi, \Lc_\pi$ given the channel state $\Sm  = \sv$, 
where the energy allocation for given $k,\ell$ and $\sv$ is deterministically given by $e_{k,\ell}(\sv) = a_\ell/s_k$. 
This provides a direct proof of the second statement of the lemma. 

In order to prove the tightness of the relaxed problem, we first show that the optimal 
$\widetilde{w}_{\ell | k}(\sv)$ is ``concentrated'' with at most two probability masses, 
at two consecutive indices $\ell, \ell+1$.   For any feasible $\widetilde{w}_k(\sv) \widetilde{w}_{\ell | k}(\sv)$, let 
\begin{eqnarray} 
E_k(\sv) & = & \sum_\ell \widetilde{w}_{\ell | k}(\sv) \frac{a_\ell}{s_k} \nonumber \\
\mu_k(\sv) & = & \sum_\ell \widetilde{w}_{\ell | k}(\sv) R_\ell.  
\end{eqnarray}
For any fixed $\sv$, by construction, the point $(E_k(\sv) , \mu_k(\sv))$ is inside the convex hull of the points
$(a_\ell/s_k, R_\ell)$ for $\ell = 0,\ldots, L$. Hence, there exists a probability mass function $\widehat{w}_{\ell | k}(\sv)$ that
has only two mass points at consecutive indices $\ell(k,\sv), \ell(k,\sv) + 1$ for some $\ell(k,\sv)$ to be determined below, 
yielding the same stability constraint and  a smaller average energy. The situation is shown qualitatively in Fig.~\ref{convex-hull}. 
The two-points probability mass function is constructed as 
follows: let $\ell(k,\sv)$ denote the index in $0, \ldots, L$ such that $R_\ell(k,\sv) \leq \mu_k(\sv) < R_{\ell(k,\sv)+1}$. 
Then, we let
\begin{equation} \label{two-mass} 
\widehat{w}_{\ell | k}(\sv) = \left \{
\begin{array}{ll}
\frac{R_{\ell(k,\sv)+1} - \mu_k(\sv)}{R_{\ell(k,\sv)+1} - R_{\ell(k,\sv)}}, & \mbox{for} \;\; \ell = \ell(k,\sv) \\
\frac{\mu_k(\sv) - R_{\ell(k,\sv)}}{R_{\ell(k,\sv)+1} - R_{\ell(k,\sv)}}, & \mbox{for} \;\; \ell = \ell(k,\sv) + 1\\
0 & \mbox{elsewhere.} \end{array} \right . 
\end{equation}
It can be immediately checked that this probability yields the average energy $\widehat{E}_k(\sv)$ given by the intersect of
the segment of straight line joining the points $(a_{\ell(k,\sv)}/s_k, R_{\ell(k,\sv)})$ and 
$(a_{\ell(k,\sv)+1}/s_k, R_{\ell(k,\sv)+1})$ with the horizontal line at level $\mu_k(\sv)$. 
Since the straight line segment is part of the boundary of the convex hull (see Fig.~\ref{convex-hull}), then 
$\widehat{E}_k(\sv) \leq E_k(\sv)$. Notice also that the union of these straight line segments coincides with the 
function $\overline{\mu}(e s_k)$ (function of the dummy energy variable $e$ on the energy-rate plane, for fixed $s_k$), where
$\overline{\mu}(\SNR)$  is the already defined upper concave envelope of 
$\mu(\SNR) = \max_\ell \{ R_\ell 1\{ \SNR \geq a_\ell \}\}$.  
 
\begin{figure}[ht]
\centerline{\includegraphics[width=12cm]{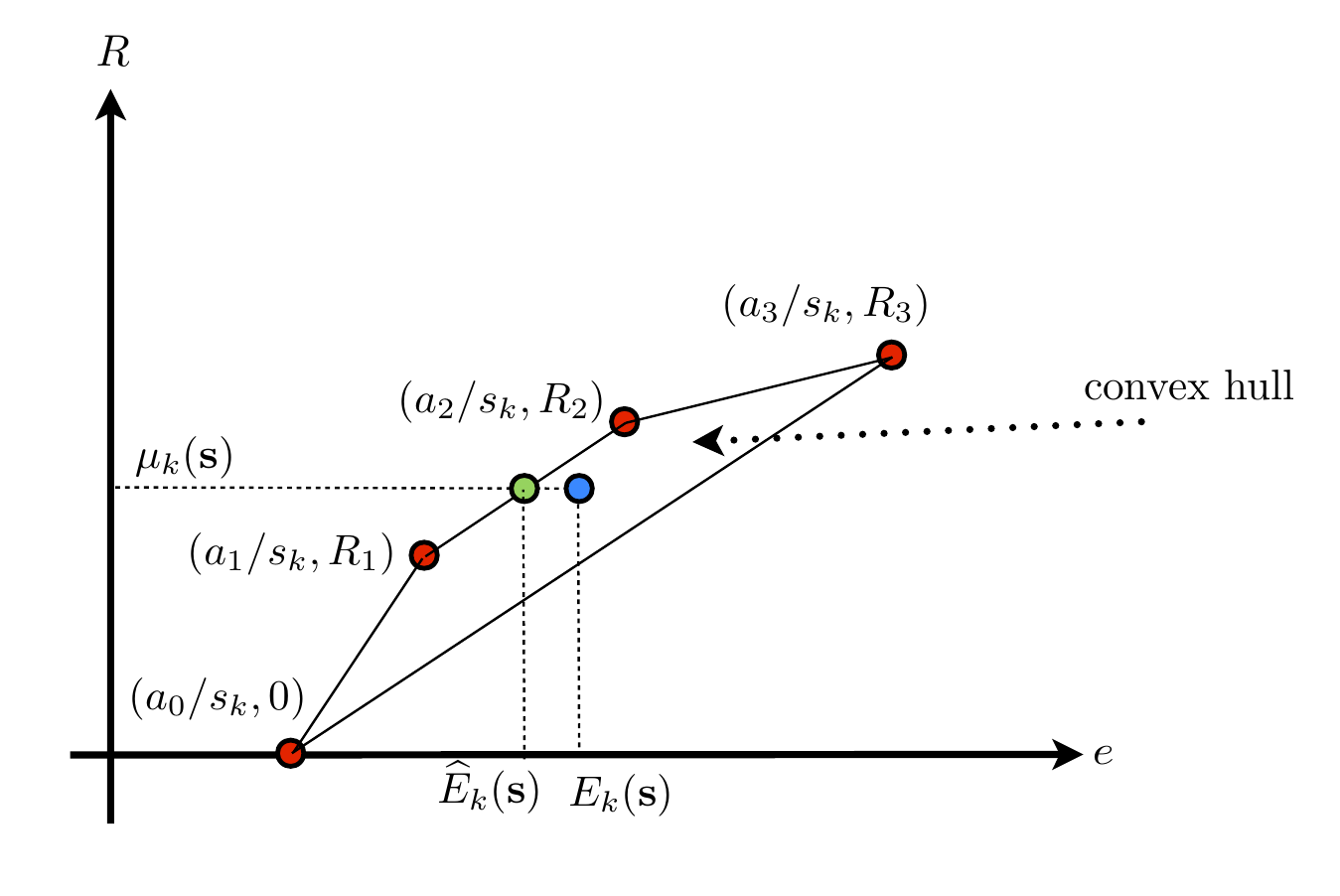}}
\caption{Convex-hull argument in the proof of Lemma \ref{indicators}: for any probability assignment achieving
the point $(E_k(\sv), \mu_k(\sv))$ (in blue) on the energy-rate plane, there exists a two-mass probability assignment
achieving the point $(\widehat{E}_k(\sv), \mu_k(\sv))$ (in green), with $\widehat{E}_k(\sv) \leq E_k(\sv)$.
In the example of the figure, this is obtained by the convex combination of the extreme points
$(a_1/s_k, R_1)$ and $(a_2/s_k, R_2)$, i.e., in this case we have $\ell(k,\sv) = 1$.}
\label{convex-hull}
\end{figure}

The proof of Lemma \ref{indicators} is concluded by observing that, for the probability assignment
$\widehat{w}_{\ell | k}(\sv)$ defined in (\ref{two-mass}), the objective function and stability constraint terms of 
the original problem (\ref{K-user-Phi-smooth1}) can be written as
\[ \EE\left [ \sum_k \widetilde{w}_k(\Sm) \widehat{E}_k(\Sm) \right ], \]
and as
\[ \EE \left [ \widetilde{w}_k(\Sm) \overline{\mu}(\widehat{E}_k(\Sm) S_k) \right ], \]
respectively. Hence, (\ref{K-user-Phi-smooth1}) is equivalent to the problem
\begin{eqnarray} 
\mbox{minimize} & &  \EE \left [ \sum_{k} \widetilde{w}_{k}(\Sm) \widehat{E}_{k}(\Sm)  \right ]  
\label{K-user-Phi-modified-indicators} \\
\mbox{subject to} & & \EE \left [ \widetilde{w}_{k}(\Sm) \overline{\mu}(\widehat{E}_{k}(\Sm) S_k ) \right ]  \geq \lambda_k, \;\;\; \forall \; k = 1, \ldots, K, \nonumber \\
& & \onev^\transp \widetilde{\wv}(\sv) = 1, \;\;\; \forall \;\; \sv \in \RR_+^K \nonumber \\
& & s_k \widehat{E}_{k}(\sv) \geq \SNR_0, \;\;\; \forall \;\; \sv \in \RR_+^K, \;\; k = 1,\ldots,K, \nonumber \\
& & \widetilde{w}_{k}(\sv)  \geq 0, \;\;\; \forall \;\; \sv \in \RR_+^K, \;\; k = 1,\ldots,K, \nonumber
\end{eqnarray}
given directly in terms of the probability vector function $\widetilde{\wv}(\sv)$ and of the energy allocation function
$\widehat{E}_k(\sv)$.  The equivalence of (\ref{K-user-Phi-modified-indicators}) with the 
the relaxed problem (\ref{K-user-Phi-smooth2}), and therefore the proof of the first statement of Lemma \ref{indicators}, is established by the following general result:

\begin{lemma} \label{equivalence}
Problem (\ref{K-user-Phi-smooth2}) yields the same value as the reduced dimensional problem
\begin{eqnarray} 
\mbox{minimize} & &  \EE \left [ \sum_{k} \widetilde{w}_{k}(\Sm) \widetilde{e}_{k}(\Sm)  \right ]  \label{K-user-Phi-modified} \\
\mbox{subject to} & & \EE \left [ \widetilde{w}_{k}(\Sm) \overline{\mu}(\widetilde{e}_{k}(\Sm) S_k ) \right ]  \geq \lambda_k, \;\;\;\forall \; k = 1, \ldots, K, \nonumber \\
& & \onev^\transp \widetilde{\wv}(\sv) = 1, \;\;\; \forall \;\; \sv \in \RR_+^K \nonumber \\
& & s_k \widetilde{e}_{k}(\sv) \geq \SNR_0, \;\;\; \forall \;\; \sv \in \RR_+^K, \;\; k = 1,\ldots,K, \nonumber \\
& & \widetilde{w}_{k}(\sv)  \geq 0, \;\;\; \forall \;\; \sv \in \RR_+^K, \;\; k = 1,\ldots,K, \nonumber
\end{eqnarray}
where $\widetilde{\wv}(\sv) = (\widetilde{w}_1(\sv), \ldots, \widetilde{w}_K(\sv))^\transp$ and
$\widetilde{\ev}(\sv) = (\widetilde{e}_1(\sv), \ldots, \widetilde{e}_K(\sv))^\transp$ are a probability and an energy allocation 
vectors, deterministic functions of the channel state. \hfill \QED
\end{lemma}

\begin{proof}
By Lemma \ref{deterministic-E}, we know that it is sufficient to consider deterministic energy allocation functions denoted by $e_{k,\ell}(\sv)$.  
For any feasible $\Wm(\sv)$ and $e_{k,\ell}(\sv)$ in (\ref{K-user-Phi-smooth2}), define
$\widetilde{w}_k(\sv) = \sum_{\ell=0}^L w_{k,\ell}(\sv)$ and write the objective function in (\ref{K-user-Phi-smooth2}) as
\begin{eqnarray} 
\EE \left [ \sum_{k,\ell} w_{k,\ell}(\Sm) e_{k, \ell}(\Sm)  \right ] 
& = & \EE \left [ \sum_{k} \widetilde{w}_k(\Sm) \sum_\ell \frac{w_{k,\ell}(\Sm)}{\widetilde{w}_k(\Sm)} e_{k, \ell}(\Sm)  \right ] \nonumber \\
& = & \EE \left [ \sum_{k} \widetilde{w}_k(\Sm) \widetilde{e}_k(\Sm) \right ] 
\end{eqnarray}
where $\widetilde{e}_k(\sv) = \sum_\ell \frac{w_{k,\ell}(\sv)}{\widetilde{w}_k(\sv)} e_{k, \ell}(\sv)$. 
Also, the terms in the rate constraints can be upper bounded as
\begin{eqnarray}
\EE \left [ \sum_{\ell} w_{k,\ell}(\Sm) \overline{\mu} (e_{k,\ell}(\Sm) S_k)   \right ] 
& = &  \EE \left [ \widetilde{w}_k(\Sm) \sum_{\ell} \frac{w_{k,\ell}(\Sm)}{\widetilde{w}_k(\Sm)} \overline{\mu} (e_{k,\ell}(\Sm) S_k)  \right ] \nonumber \\
& \leq &  \EE \left [ \widetilde{w}_k(\Sm) \overline{\mu} \left ( \sum_{\ell} \frac{w_{k,\ell}(\Sm)}{\widetilde{w}_k(\Sm)}  e_{k,\ell}(\Sm) S_k \right )  \right ] \nonumber \\
& = &  \EE \left [ \widetilde{w}_k(\Sm) \overline{\mu} \left ( \widetilde{e}_{k}(\Sm) S_k \right )  \right ],  
\end{eqnarray}
where we used Jensen's inequality.
This shows that the value of problem (\ref{K-user-Phi-modified}) lower bounds the value of problem (\ref{K-user-Phi-smooth2}).
Conversely, problem (\ref{K-user-Phi-modified}) corresponds to a specific feasible choice of 
$\Wm(\sv)$ and $e_{k,\ell}(\sv)$ in the original problem (\ref{K-user-Phi-smooth2}), 
in particular,  by letting $w_{k,\ell}(\sv) = \widetilde{w}_k(\sv) /(L+1)$ and $e_{k,\ell}(\sv) = \widetilde{e}_k(\sv)$, 
for all $\ell = 0,\ldots, L$, $k = 1,\ldots, K$ and $\sv \in \RR_+^K$. 
This concludes the proof.
\end{proof}

\newpage

\bibliographystyle{IEEEtran}
\bibliography{body-area}

\clearpage

\begin{table}[t]
\begin{center}
\begin{tabular}{l|l}
\hline \hline
Sensor type & Throughput (kb/s) \\ \hline
SpO2 & 0.01 -- 0.1 \\
Glucose & 0.01 -- 0.1 \\
Blood pressure & 0.01 -- 10\\
ECG & 10 -- 100 (12-bit, 300 Hz, $\times$20) \\
EEG & 10 -- 200 (6kb/s, $\times$32) \\
EMG & 10 -- 1500 (16-bit, 8 kHz, $\times$12)
\end{tabular}
\end{center}
\caption{Throughput requirements for medical sensing devices}
\label{table}
\end{table}

\begin{table}[t]
\begin{center}
\begin{tabular}{l|l|l|l}
\hline \hline
Raw Data \\ Rate (Mbps) & Slots Occupied &  No. bits in payload & Notation \\ \hline
2 & 1 & 464 & 2dh1\\
2 & 3 & 2968 & 2dh3 \\
2 & 5 & 5464 & 2dh5 \\
3 & 1 & 696 & 3dh1\\
3 & 3 & 4448 & 3dh3 \\
3 & 5 & 8200 & 3dh5 \\
\end{tabular}
\end{center}
\caption{Payload for each of the enhanced data rate modes for different time slots. The indication xdhy denotes
an {\em uncoded} transmission mode with raw spectral efficiency of x bit/s/Hz, and payload spanning y slots.}
\label{PacketSlots}
\end{table}

\begin{figure}[ht]
\centerline{\includegraphics[width=12cm]{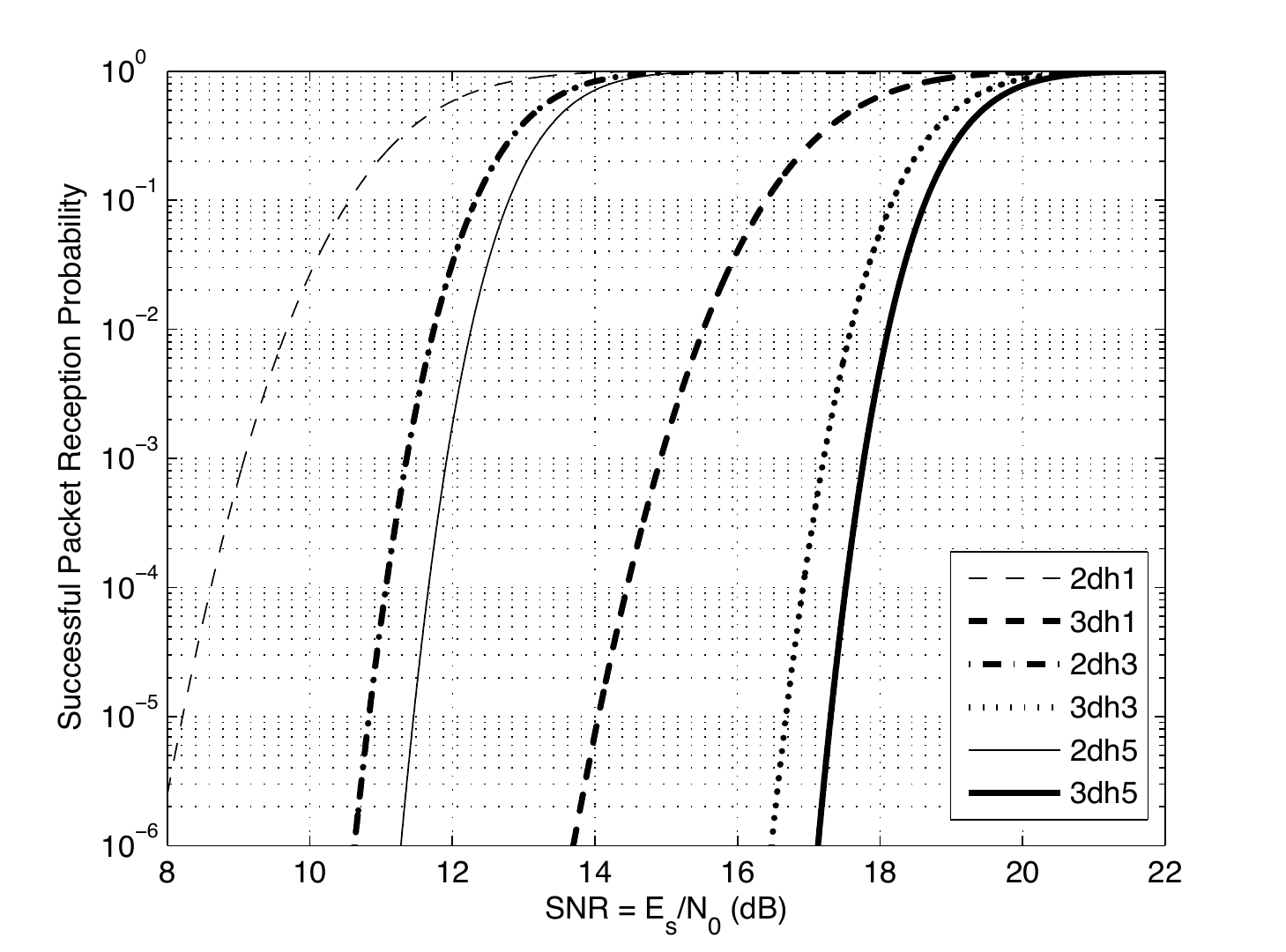}}
\caption{Packet success probability for enhanced data rates EDR2 and EDR3 for all packet lengths under AWGN.}
\label{PERfig}
\end{figure}

\begin{figure}[ht]
\centerline{\includegraphics[width=12cm]{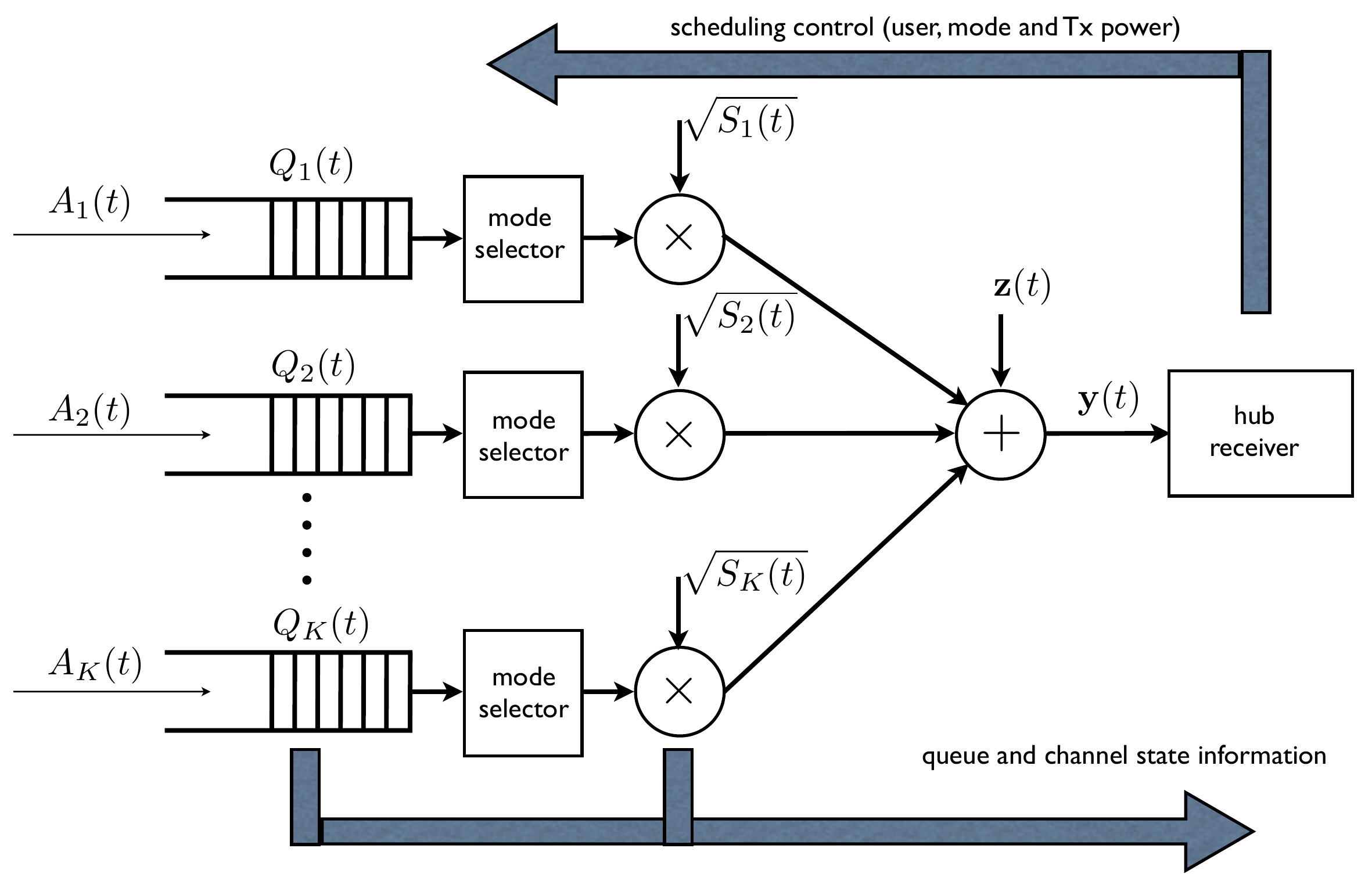}}
\caption{Block diagram of the system under consideration consisting of $K$ sensors with arrival rates $A_k(t)$, each equipped with 
buffer of length $Q_k(t)$ at time $t \in (0,\infty)$, $k \in [1,...,K]$, where each of the sensors experiences independent fading realizations $S_k(t)$.}
\label{system1}
\end{figure}

\begin{figure}[ht]
\centerline{\includegraphics[width=12cm]{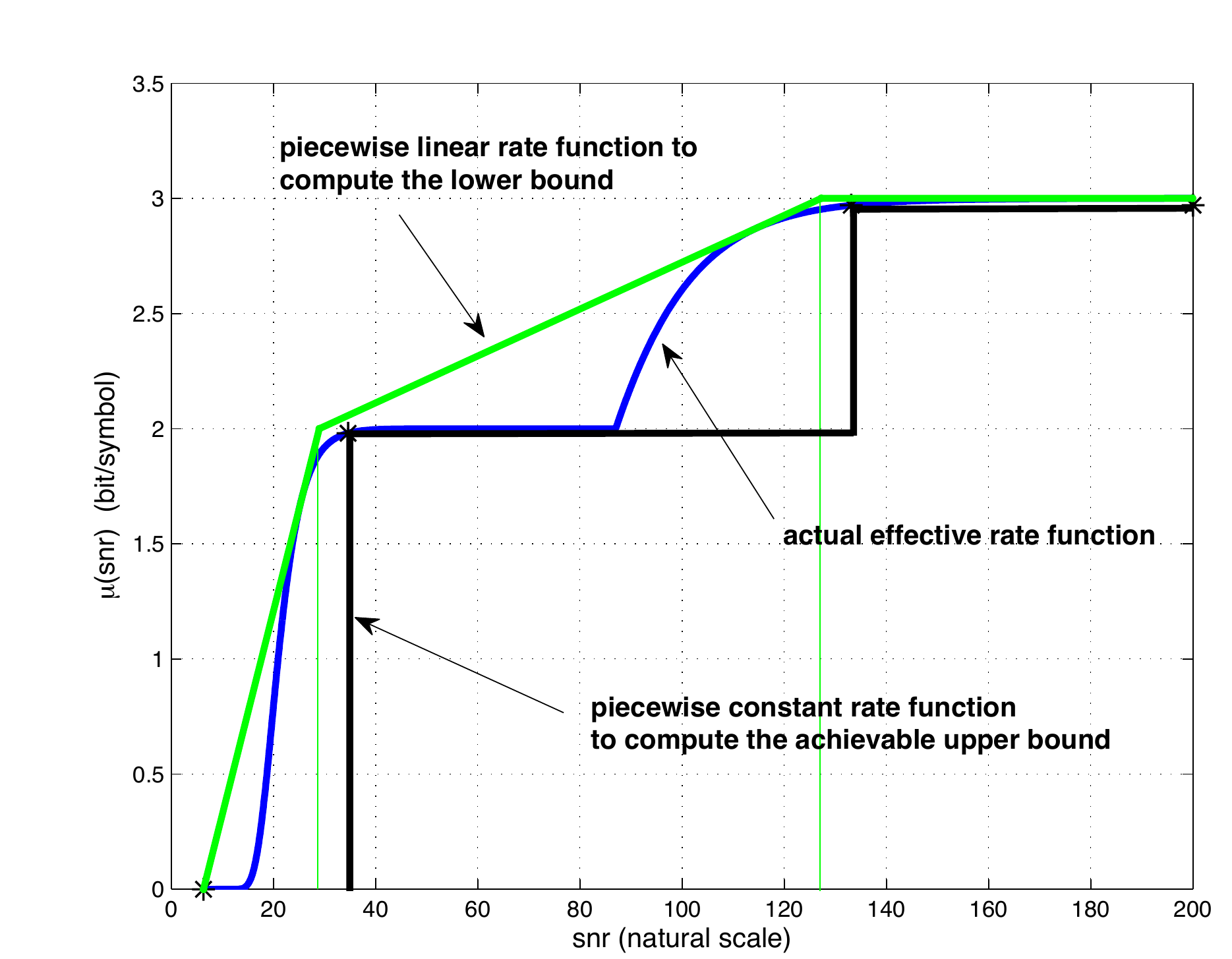}}
\caption{Effective rate function $\mu(\SNR)$, the piecewise constant function (maximum of indicator functions) 
corresponding to the computation of the energy function achievable upper bound $\overline{\Phi}(\lambdav)$ and the piecewise 
linear function $\widetilde{\mu}(\SNR)$ corresponding to the computation of the lower bound $\underline{\Phi}(\lambdav)$, 
for the case of Bluetooth with modes 2dh3 and 3dh3, with $\SNR_0 = 8$ dB. In this chart, the SNR axis is in linear scale in order to 
appreciate the actual shape of the curves appearing in the optimization problems.}
\label{PERfig-example}
\end{figure}

\begin{figure}[ht]
\centerline{\includegraphics[width=14cm]{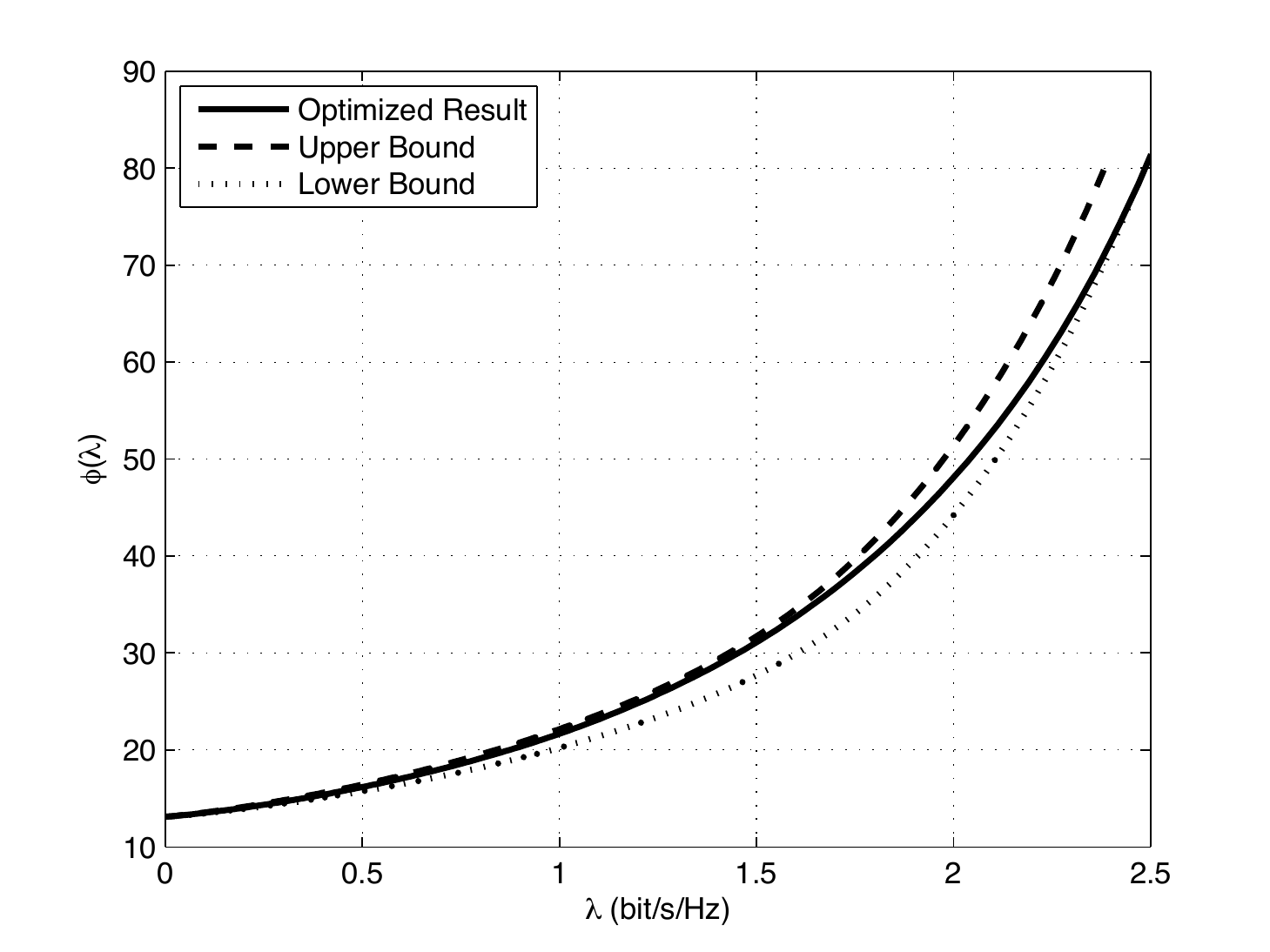}}
\caption{Functions $\underline{\Phi}(\lambda)$ and $\overline{\Phi}(\lambda)$ (shown in natural scale, not in dB) 
for the case of $K = 1$, with  truncated Ricean channel state statistics, and rates corresponding to Fig.~\ref{PERfig-example}.}
\label{bounds}
\end{figure}

\begin{figure}[ht]
\centerline{\includegraphics[width=12cm]{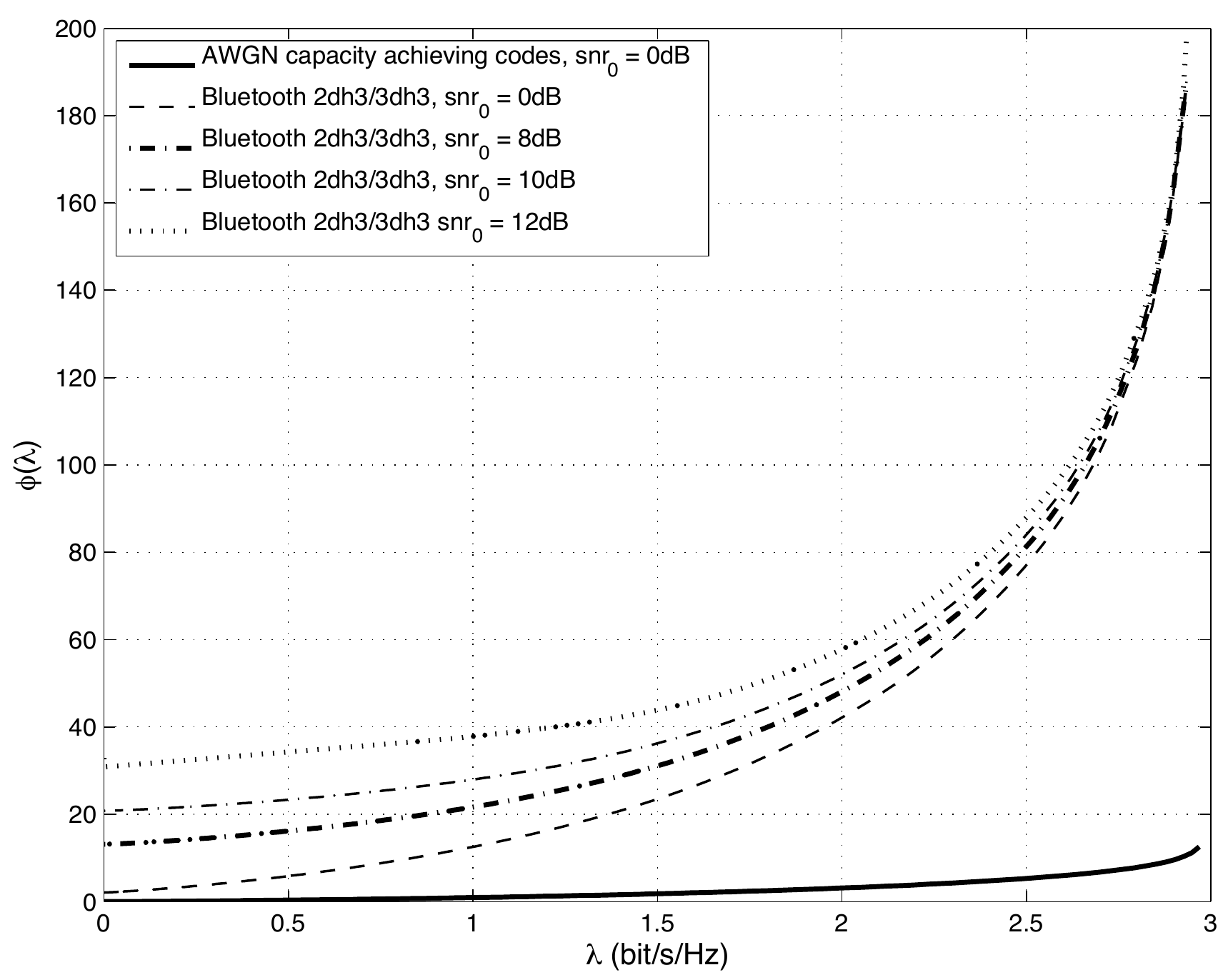}}
\caption{Average sum energy function $\overline{\Phi}(\lambda)$ versus arrival rate $\lambda$ for different $\SNR_0$ values for the 
Bluetooth-like protocol with PHY modes 2dh3 and 3dh3 with $K = 1$ sensors. The corresponding minimum energy function $\Phi(\lambda)$ for an 
ideal system with capacity-achieving codes of rate 2 and 3 bit/s/Hz, and  NULL packet SNR 
requirement $\SNR_0 = 0$ dB is given for comparison.}
\label{sec5}
\end{figure}

\begin{figure}[ht]
\centerline{\includegraphics[width=12cm]{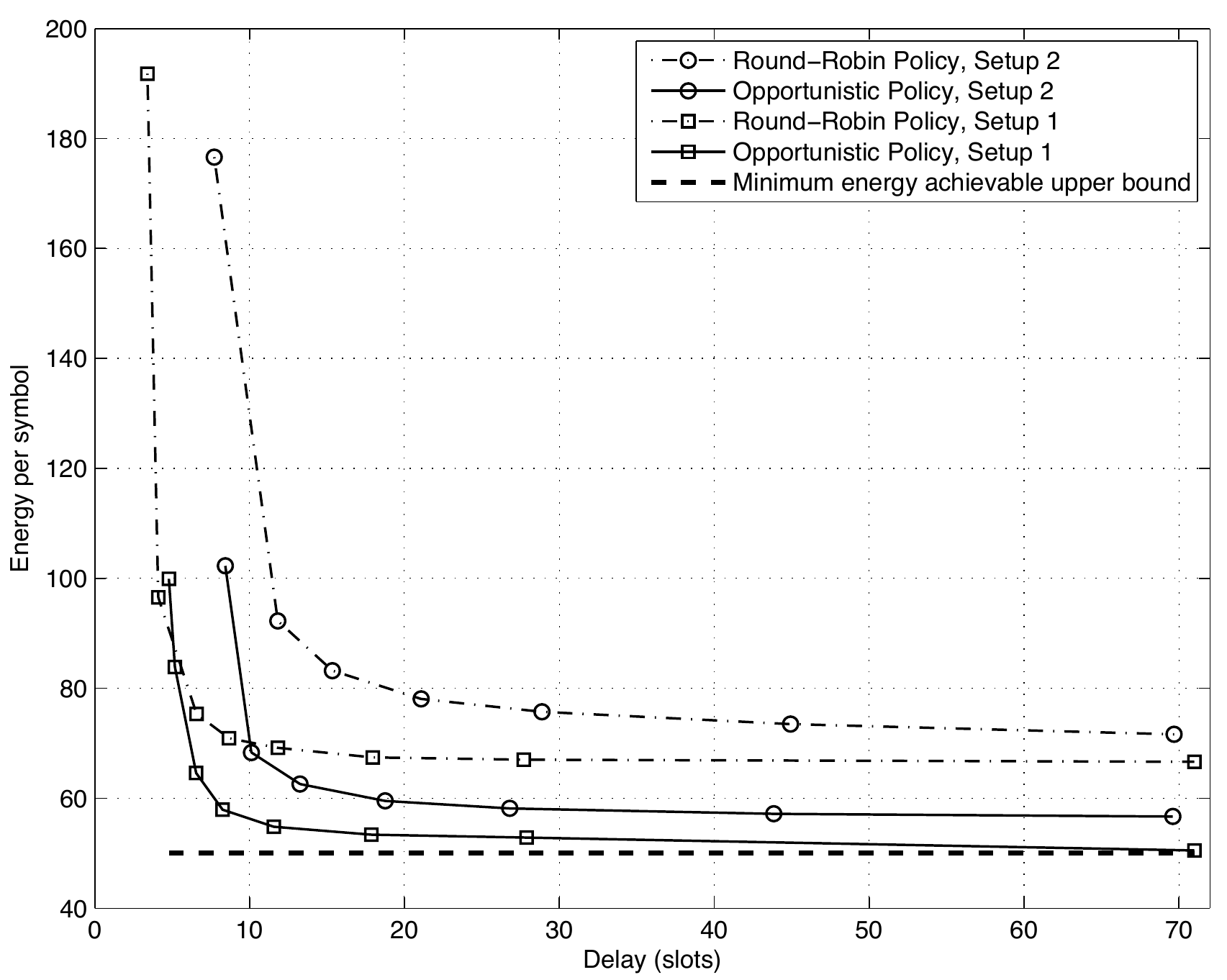}}
\caption{Comparison of round-robin vs. multiuser dynamic scheduling policy (opportunistic) for a piconet with $K = 2$ sensors, with symmetric arrival rates 
$\lambda_1 = \lambda_2 =1$ bit/s/Hz, and $\SNR_0 = 8$ dB, under different arrival statistics. 
Setup 1 considers $q_1 = 1, q_2 = 1$, whereas Setup 2 considers $q_1 = 0.2, q_2 = 1$.
The horizontal line corresponds to the achievable upper bound to the minimum energy function $\overline{\Phi}(\lambdav)$.}
\label{opportunistic}
\end{figure}

\begin{figure}[ht]
\centerline{\includegraphics[width=12cm]{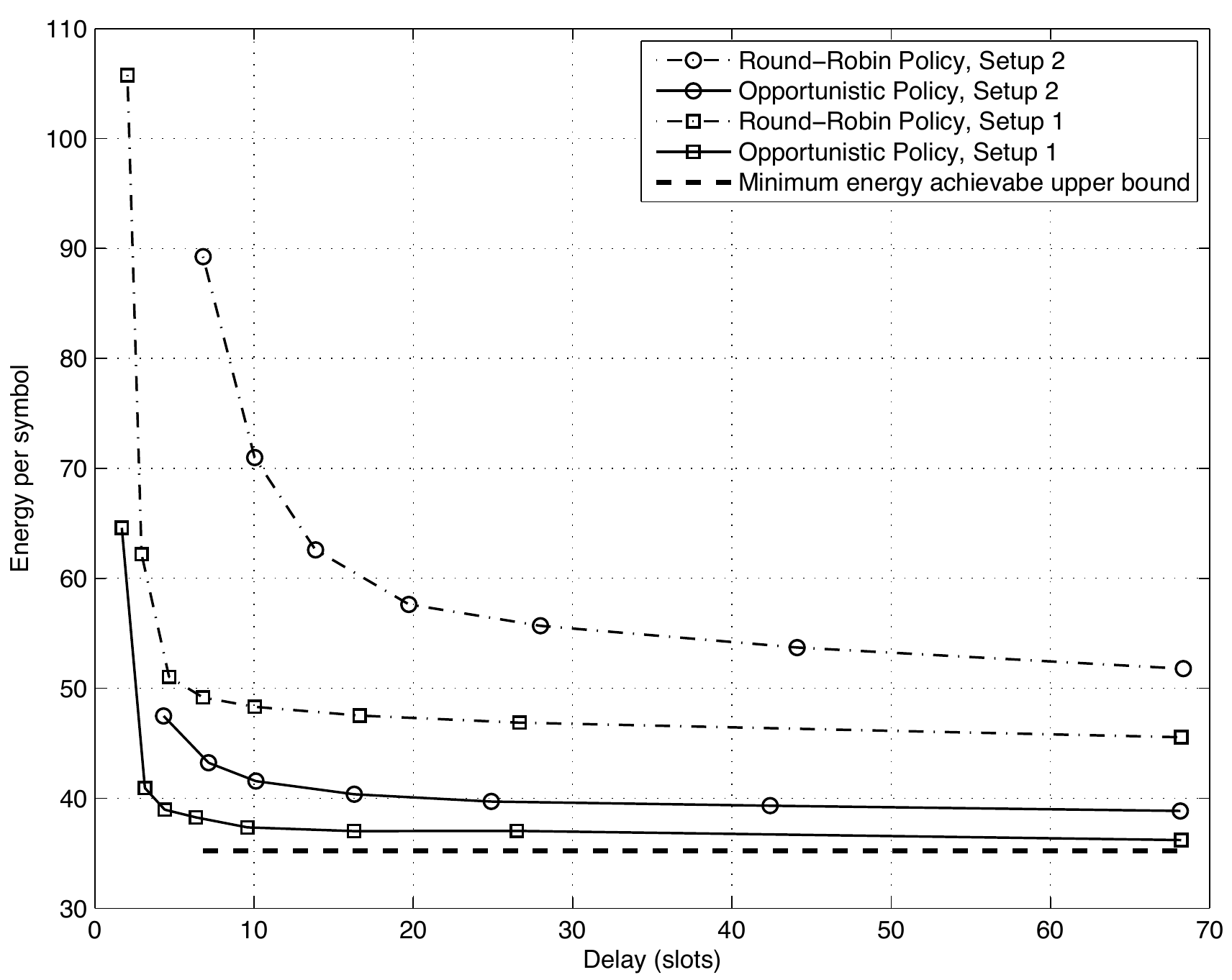}}
\caption{Comparison of round-robin vs. multiuser dynamic scheduling policy (opportunistic)   
with asymmetric arrival rates $\lambda_1 = 0.04$ bit/s/Hz, $\lambda_2=1$ bit/s/Hz.
All other parameters are as in Fig.~\ref{opportunistic}.
The horizontal line corresponds to the achievable upper bound to the minimum energy function $\overline{\Phi}(\lambdav)$.}
\label{opportunistic2}
\end{figure}

\begin{figure}[ht]
\centerline{\includegraphics[width=12cm]{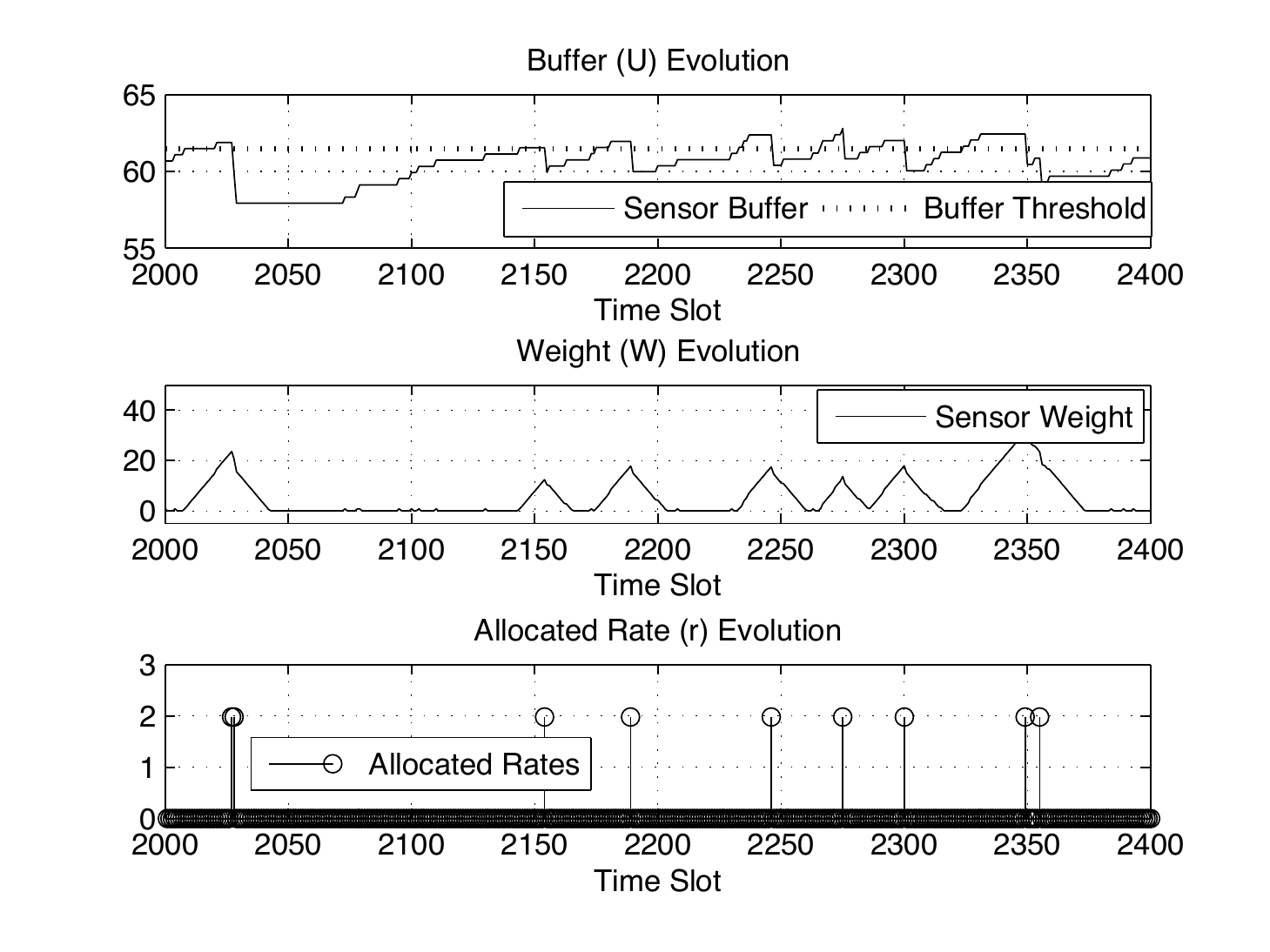}}
\caption{Snapshot of the evolution of the buffer $Q(t)$, scheduler weight $W(t)$ 
and scheduled rate for a bursty arrival process with $\lambda = 0.04$ bit/s/Hz and $q = 0.1$.}
\label{kazzi}
\end{figure}

\begin{figure}[ht]
\centerline{\includegraphics[width=14cm]{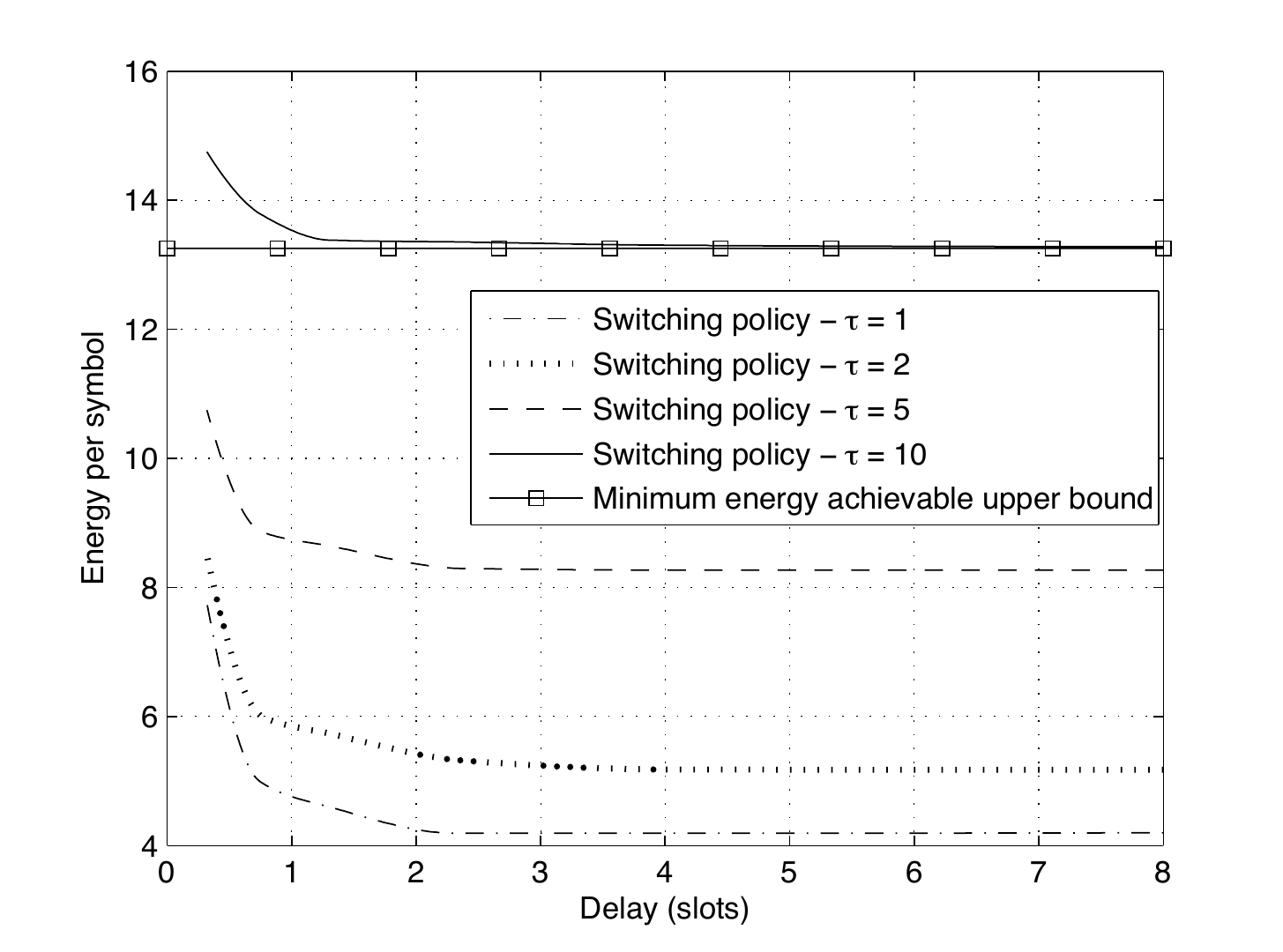}}
\caption{Average energy/delay tradeoff achieved by the improved scheduling scheme with adaptive switching to sleeping mode and 
a bursty arrival process with $\lambda = 0.04$ bit/s/Hz, $q=0.1$ and $\SNR_0=8$ dB, for different $\tau$ factors. 
The case $\tau = 10$ 
corresponds to  the classical non-switching policy.}
\label{switch-fig}
\end{figure}

\begin{figure}[ht]
\centerline{\includegraphics[width=12cm]{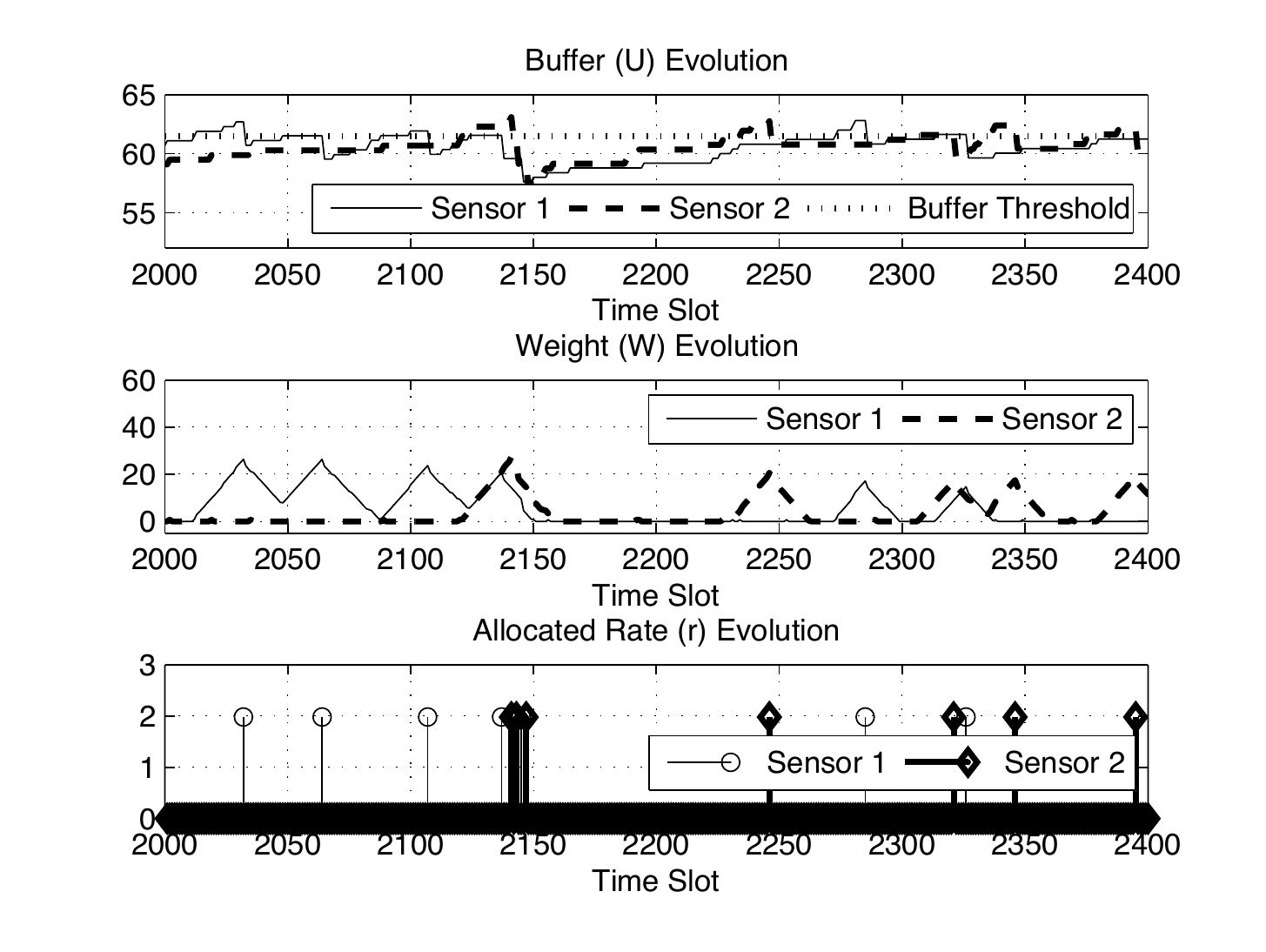}}
\caption{Snapshot of the evolution of the buffers $Q_k(t)$, scheduler weights $W_k(t)$ 
and scheduled rates for $K = 2$, independent bursty arrival processes with 
$\lambda_1 = \lambda_2 = 0.04$ bit/s/Hz and $q_1 = q_2 = 0.1$.}
\label{kazzi2}
\end{figure}

\begin{figure}[ht]
\centerline{\includegraphics[width=12cm]{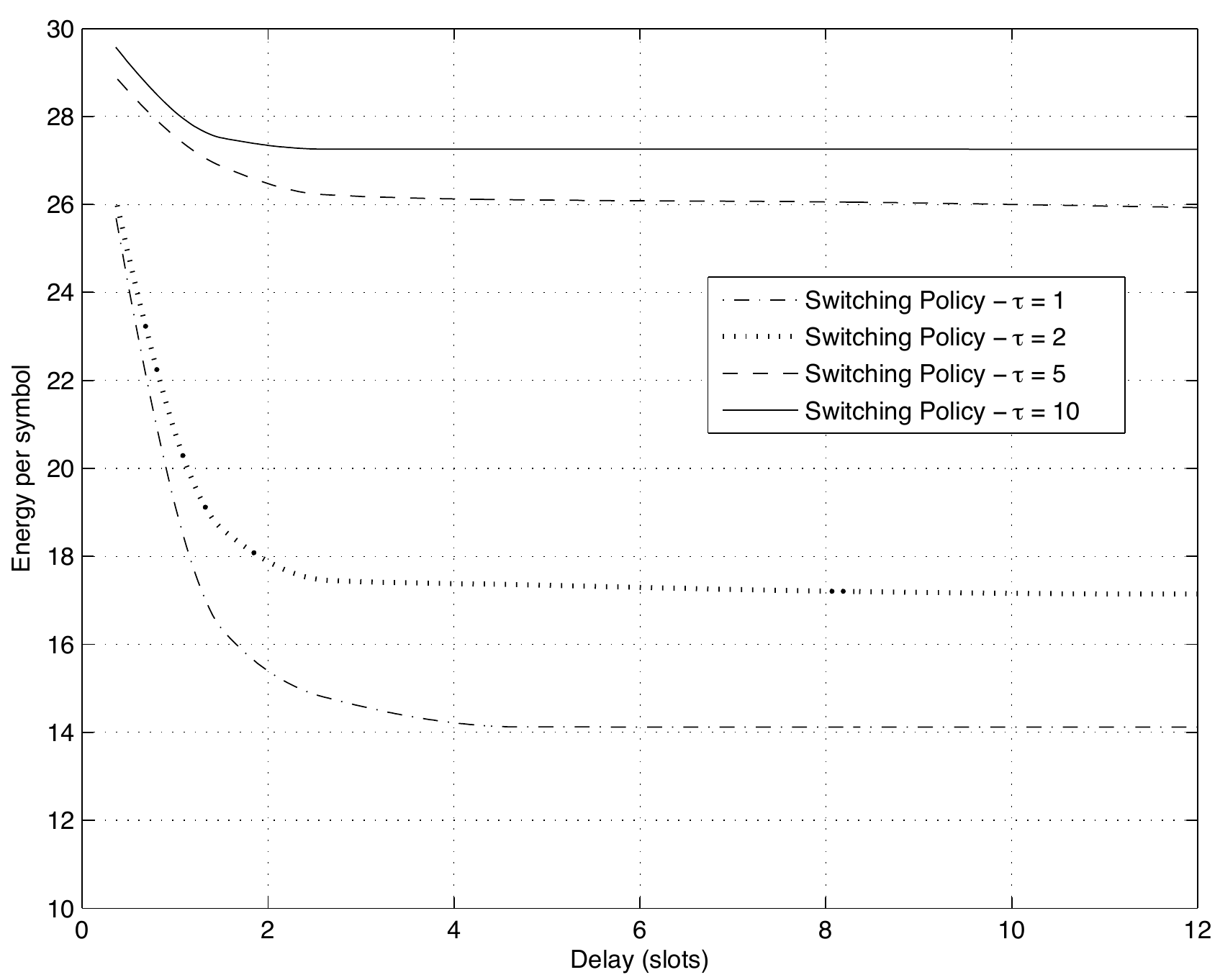}}
\caption{Average energy/delay tradeoff achieved by the improved scheduling scheme with adaptive switching to sleeping mode 
with $K = 2$,  independent bursty arrival processes with $\lambda_1 = \lambda_2 = 0.04$ bit/s/Hz, $q_1 = q_2 = 0.1$
and $\SNR_0=8$ dB. The case $\tau = 10$ corresponds to  the classical non-switching policy.}
\label{switch-fig2}
\end{figure}

\begin{figure}[ht]
\centerline{\includegraphics[width=12cm]{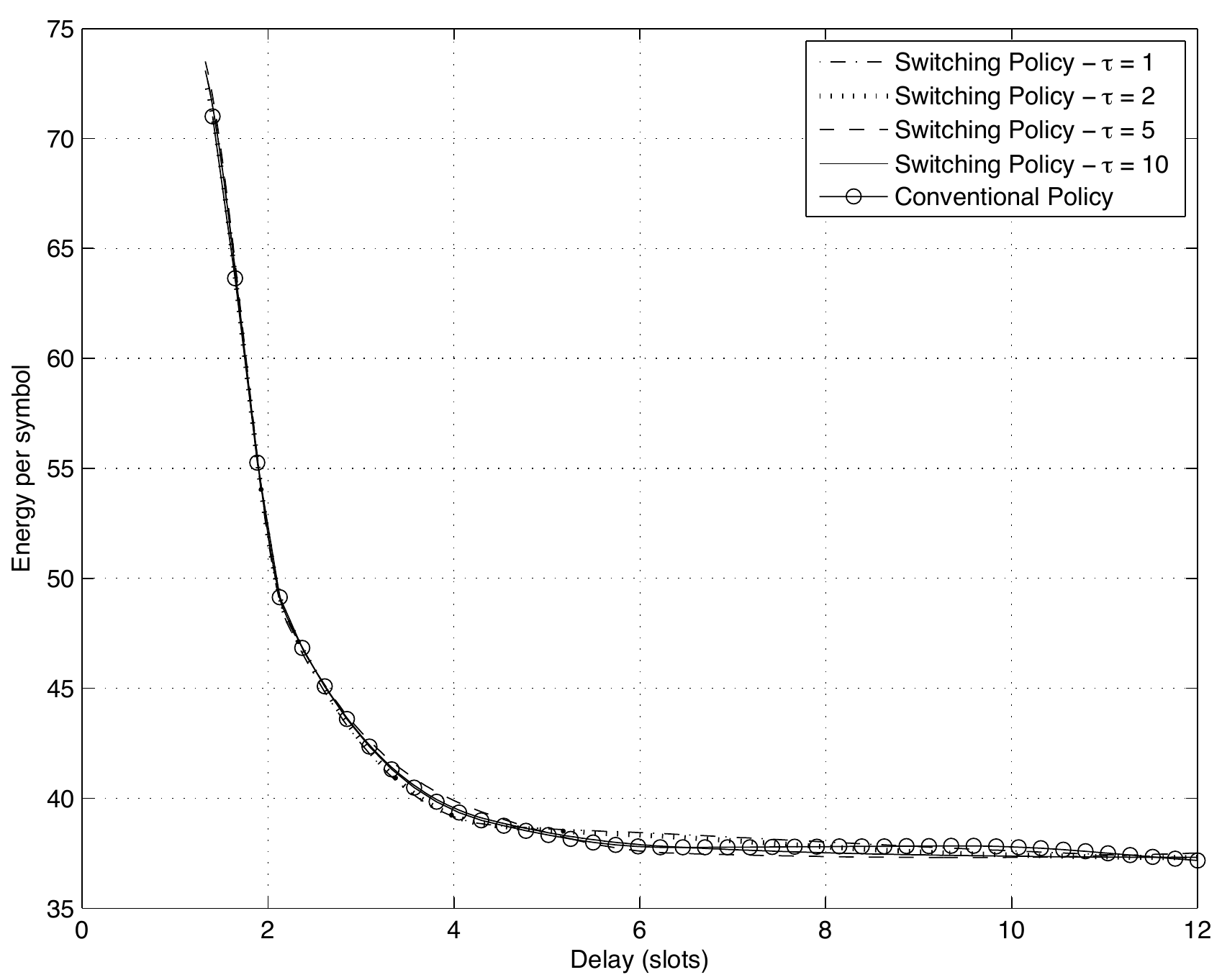}}
\caption{Average energy/delay tradeoff achieved by the improved scheduling scheme with adaptive switching to sleeping mode 
for $K = 2$,  independent arrival processes with $\lambda_1 = 0.04$ bit/s/Hz, $q_1 = 0.1$, 
$\lambda_2 = 1.0$ bit/s/Hz and $q_2 = 1.0$  and $\SNR_0=8$ dB.}
\label{switch-fig3}
\end{figure}

\end{document}